\documentclass[12pt]{article}
\usepackage[letterpaper, left=1.5cm, right=1.5cm, top=2cm, bottom=2.5cm]{geometry}
\usepackage{graphicx} 
\usepackage{subcaption, caption}
\usepackage{float}
\RequirePackage{amsthm,amsmath,amsfonts,amssymb, bbm}
\usepackage{mathtools}
\usepackage{xcolor}
\usepackage{hyperref}
\hypersetup{pageanchor=false}
\usepackage[authoryear,round,longnamesfirst]{natbib}
\usepackage{comment}
\usepackage{authblk} 

\definecolor{blue}{RGB}{0,0,255}
\definecolor{red}{RGB}{255,0,0}
\definecolor{purple}{RGB}{102,0,204}

\def\spacingset#1{\renewcommand{\baselinestretch}
{#1}\small\normalsize} \spacingset{1}

\DeclareMathOperator*{\argmin}{arg\,min}

\newtheorem{proposition}{Proposition}
\newtheorem{lemma}{Lemma}

\newcommand{\indep}{\perp \!\!\! \perp}

\newcommand{\cov}{\mbox{Cov}}
\newcommand{\Cov}{\mbox{Cov}}

\newcommand{\dCov}{\mbox{dCov}} 
\newcommand{\dVar}{\mbox{dVar}}
\newcommand{\dStd}{\mbox{dStd}}
\newcommand{\dCor}{\mbox{dCor}}
\newcommand{\E}{\mathbb{E}}
\newcommand{\bX}{\boldsymbol{X}}
\newcommand{\bXn}{\boldsymbol{X_n}}
\newcommand{\bYn}{\boldsymbol{Y_n}}
\newcommand{\bZn}{\boldsymbol{Z_n}}
\newcommand{\bP}{\boldsymbol{P}}
\newcommand{\bA}{\boldsymbol{A}}
\newcommand{\bhA}{\boldsymbol{\widehat{A}}}
\newcommand{\bSn}{\boldsymbol{S_n}}
\newcommand{\bSnk}{\boldsymbol{S_{n,k}}}
\newcommand{\bSnl}{\boldsymbol{S_{n,\ell}}}
\newcommand{\bSnlplus}{\boldsymbol{S_{n,\ell_{+}}}}
\newcommand{\bSnkplus}{\boldsymbol{S_{n,k_{+}}}}
\newcommand{\bI}{\boldsymbol{I}}
\newcommand{\calJ}{\mathcal{J}}

\newcommand{\bQ}{\boldsymbol{Q}}
\newcommand{\bS}{\boldsymbol{S}}
\newcommand{\bs}{\boldsymbol{s}}
\newcommand{\bU}{\boldsymbol{U}}
\newcommand{\bUo}{\boldsymbol{U_0}}
\newcommand{\bUh}{\boldsymbol{\widehat{U}}}
\newcommand{\bW}{\boldsymbol{W}}

\newcommand{\bY}{\boldsymbol{Y}}
\newcommand{\bZ}{\boldsymbol{Z}}
\newcommand{\bx}{\boldsymbol{x}}
\newcommand{\by}{\boldsymbol{y}}

\newcommand{\bhSigma}{\boldsymbol{\widehat{\Sigma}}}

\newcommand{\btheta}{\boldsymbol{\theta}}
\newcommand{\bphi}{\boldsymbol{\phi}}
\newcommand{\bhtheta}{\boldsymbol{\widehat{\theta}}}
\newcommand{\Op}{\mathcal{O}_{\text{P}}}
\newcommand{\op}{o_{\text{P}}}

\newcommand{\biloop}{\mathbin{\rotatebox[origin=c]{-90}{$\scriptscriptstyle0$}\rotatebox[origin=c]{-90}{$\scriptscriptstyle0$}}} 

\title{Independent Component Analysis\\ by
       Robust Distance Correlation}

\author[1]{Sarah Leyder}
\author[1]{Jakob Raymaekers}
\author[2]{Peter J. Rousseeuw}
\author[1]{\\Tom Van Deuren}
\author[1]{Tim Verdonck}

\affil[1]{Department of Mathematics, University of 
          Antwerp, Belgium}
\affil[2]{Section of Statistics and Data Science, Department 
          of Mathematics, KU Leuven, Belgium}
\date{May 14, 2025} 

\begin{document}

\hyphenation{FastICA}
\hyphenation{biloop}
\hyphenation{RICA}
\hyphenation{dCovICA}

\spacingset{1.5}

\maketitle

\large

\abstract{Independent component analysis (ICA) 
is a powerful tool for decomposing a multivariate
signal or distribution into truly independent 
sources, not just uncorrelated ones. Unfortunately, 
most approaches to ICA are not robust against 
outliers. Here we propose a robust ICA method 
called RICA, which estimates the components by 
minimizing a robust measure of dependence 
between multivariate random variables.
The dependence measure used is the distance
correlation (dCor). In order to make it
more robust we first apply a new transformation 
called the bowl transform, which is bounded,
one-to-one, continuous, and maps far outliers 
to points close to the origin. 
This preserves the crucial property that 
a zero dCor implies independence.
RICA estimates the independent
sources sequentially, by looking for the
component that has the smallest dCor with
the remainder. RICA is strongly consistent and 
has the usual parametric rate of convergence.
Its robustness is investigated by a simulation 
study, in which it generally outperforms
its competitors. The method is illustrated on 
three applications, including the well-known 
cocktail party problem.
}

\section{Introduction}
\label{sec:introduction}

Independent component analysis (ICA) is a popular and 
powerful technique in statistical signal processing,
originating from the work of \cite{Ans1985} 
and \cite{comon1994independent}; for a review see
the book of \cite{Hyv2001book}. ICA bears a
superficial resemblance to principal component
analysis (PCA), but there are important differences.
The goal of PCA is to rotate the coordinate axes
to obtain a coordinate system in which the variables 
become uncorrelated. The new variables are then
called principal components, after which one may 
reduce the data dimension by keeping only the 
principal components with the highest variance.
On the other hand, ICA looks for a coordinate
system in which the variables, called {\it sources},
are independent. This is quite different 
because uncorrelated variables are not necessarily 
independent, so for ICA it does not suffice to know 
the covariance matrix of the data. 
Another crucial difference is that ICA allows to 
transform the coordinate axes by a general linear 
transformation, whereas PCA was restricted to 
orthogonal transformations.
Because the dependence structure of a multivariate
Gaussian distribution is fully determined by its 
nonsingular covariance matrix, which can be turned 
into any other positive definite matrix by a linear 
data transformation, ICA can only detect 
non-Gaussian independent components.

Taken together, the above aspects make 
ICA more complex than PCA. But it has many 
important applications in statistics and 
engineering. Several methodological tools depend 
on ICA, including blind source separation 
\citep{cardoso1998,comon2010handbook}, 
causal discovery \citep{shimizu2006linear}, 
feature extraction, preprocessing, 
artifact removal, and noise reduction. 
Tools based on ICA have been applied 
across various fields 
such as biomedical sciences, audio and image 
preprocessing, biometrics, and finance. For 
application overviews see e.g. 
\cite{tharwat2021ica} 
and \cite{naik2011overview}. In such 
applications the independent sources often have 
natural interpretations. In 
Section~\ref{sec:real data examples} we will 
illustrate this with
image data, audio data, and periodic data.

The basic formulation of the ICA setting goes as 
follows. Suppose we observe $n$ i.i.d. observations 
$\bx_1, \ldots, \bx_n$ which are realizations of a 
$d$-dimensional random vector $\bX$. We denote the
$n \times d$ data matrix by $\bX_n$. 
The independent component model assumes that the
random vector $\bX$ can be written as
\begin{equation} \label{eq:X=AS}
    \bX = \bA \bS
\end{equation}
where $\bS$ is an unobserved $d$-variate 
non-Gaussian random vector whose components are 
mutually independent. Here $\bA$ is an unknown 
non-singular $d \times d$ matrix, called the 
\textit{mixing matrix}. Only the linear mixture 
$\bX$ is observable, whereas $\bA$ and $\bS$ are
both unknown. Given the data $\bx_1,\ldots,\bx_n$\,, 
the goal of ICA is to estimate the mixing matrix 
$\bhA$ which can then be used to ``unmix'' the 
observed $\bx_i$ by computing 
$\bs_i = \bhA^{\,-1} \bx_i$ for 
$i = 1, \ldots, n$. Note that the scale, sign and 
order of the components of $\bS$ are not identifiable 
from the distribution of $\bX$. Therefore we aim to 
recover $\bA$ up to multiplication by a signed 
permutation matrix and a diagonal scale matrix.

Several generalizations of the model \eqref{eq:X=AS} 
exist. For instance, some consider a ``noisy'' version
which adds some extra Gaussian noise to $\bX$, 
see e.g. \cite{voss2015}. Other 
generalizations allow some of the components in $\bS$ 
to be Gaussian, and target the identification of the 
non-Gaussian components only
\citep{nordhausen2018independent}.

Various approaches exist for separating the mixed 
sources that created $\bX$. A first set of algorithms 
uses measures of non-Gaussianity to extract $\bS$, 
such as kurtosis or approximations of negentropy
\citep{hyvarinen1997original,hyvarinen1999robust},
resulting in 
the renowned FastICA algorithm of 
\cite{hyvarinen2000independent}. A second approach 
for extracting the independent sources is 
to maximize entropy, as done by the infomax 
\citep{bell1995infomax} and extended infomax 
\citep{lee1999extendedinfomax} methods. A third 
approach for recovering the independent sources is
to jointly diagonalize fourth order cumulants, 
resulting in the JADE algorithm of 
\cite{cardoso1993jade,miettinen2015fourth}.

Although these ICA methods have demonstrated success 
and broad applicability, they generally exhibit 
a high sensitivity to outliers. One reason is that  
classical statistical estimators of kurtosis and 
higher-order cumulants are known to be heavily 
influenced by outliers. The effect of outliers on
an estimator can be measured by its influence 
function \citep{robstatIF86}, which ideally should 
be bounded. 
But FastICA has an unbounded influence function,  
regardless of the chosen contrast function 
\citep{hyvarinen1999robust,ollila2009IFfastica,
nordhausen2011deflation}. 
The issue is that outliers can generate new 
dependencies that violate the underlying independent 
component model. Therefore we cannot expect to recover 
the sources perfectly, but we can try to construct a
more robust ICA method that mitigates the bias
caused by outliers as much as possible. 

Some studies have addressed the robustness of ICA 
algorithms. \cite{hyvarinen1999robust} analyzed 
the statistical properties of different contrast 
functions within the FastICA algorithm, leading to 
recommendations to achieve more robustness. While 
kurtosis is very sensitive to outliers, he
proposed bounded alternative contrast functions, 
thereby reducing the effect of extreme values.
These findings were supported by a simulation study. 
Building on this, \cite{brys2005robustification} 
extended FastICA by adding a preprocessing 
step using an outlier rejection rule. This removes 
gross outliers, but makes the method less efficient
when the data is not contaminated by outliers. 
\cite{RADICAL} proposed the RADICAL method based on 
a nonparametric entropy estimator. While certainly 
more robust than FastICA, the robustness of RADICAL 
deteriorates substantially with increasing 
dimension and number of outliers. 
\cite{chen2013gammaica} proposed a robust approach to 
ICA using $\gamma$-divergence. Their method, termed 
$\gamma$-ICA, aims to leverage the robustness 
properties of $\gamma$-divergence against 
outliers, and is consistent when the components of 
$\bS$ are symmetrically distributed. However, 
$\gamma$-ICA requires the selection of the tuning 
parameter $\gamma$, and the empirical evaluation 
of its robustness was somewhat limited.

In this work we propose a new robust ICA method. 
It is based on the dCovICA method of 
\cite{dCovICAmatteson2017}, which minimizes the 
dependence between the recovered signals. 
For this purpose it measures dependence by the 
distance covariance of \cite{szekely2007dcor}. 
Although this measure of dependence offers some 
inherent resistance against outliers, it was 
recently shown in \citep{leyder2024distance} that 
its robustness against increasing numbers of gross 
outliers can be improved. This was achieved by 
subjecting the data to a new type of transformation
before computing the distance covariance. In the 
current paper we extend this approach from 
univariate to multivariate variables. We then 
use it to construct RICA, a \textbf{R}obust 
\textbf{ICA} procedure, which operates by 
minimizing this robust dependence measure 
between a potential new source and the 
remainder. 

The paper is structured as 
follows. Section~\ref{sec:prelim} introduces some
preliminary concepts needed to construct RICA. 
Section~\ref{sec:methodology} provides a detailed 
description of the RICA methodology. Next, 
Section~\ref{sec:theoretical analysis} formulates
theoretical properties of the method. The 
extensive simulation results presented in 
Section~\ref{sec:simulations} demonstrate the 
enhanced robustness of RICA against outliers. 
Section~\ref{sec:real data examples} showcases 
some examples on real datasets, and  
Section~\ref{sec:discussion} concludes.

\section{Preliminaries}
\label{sec:prelim}
We first introduce the background concepts on 
which we build our proposal. 

\subsection{Whitening}
\label{subsec:whitening}
The problem of identifying the {\it unmixing matrix}
$\bA^{-1}$ in the ICA problem can be simplified 
somewhat by first decorrelating the observed data 
$\bX_n$\,. This process is known as ``whitening''. 
For this one typically constructs a nonsingular
$d \times d$ matrix $\bW$ such that $\bZ = \bW \bX$ 
has uncorrelated components with unit variances, 
i.e. $\cov(\bZ) = \bI$. Then we can write
\begin{equation*}
    \bS = \bA^{-1}\bX = \bA^{-1}\bW^{-1}\bZ = \bU \bZ,
\end{equation*}
where $\bU = \bA^{-1}\bW^{-1}$ is called the 
{\it separating matrix}. Let us assume without loss
of generality that $\Cov(\bS) = \bI$. Then the 
separating matrix is an orthogonal matrix because
\begin{equation*}
   \bI = \cov(\bS) = \bU \cov(\bZ) \bU^T = \bU \bU^T.
\end{equation*}
Orthogonal $d \times d$ matrices only have 
$d(d-1)/2$ free elements instead of $d^2$, as will
be described in Section~\ref{subsec:optimization}.
Therefore whitening reduces the dimension of 
the search space for $\bU$ to roughly half of 
the size of the unmixing matrix $\bA^{-1}$.

\subsection{Distance covariance and correlation}
\label{subsec:dcovdcor}

The distance covariance (dCov) was introduced 
by \cite{szekely2007dcor} as a measure of general 
dependence between random vectors. For random 
vectors $\bX$ and $\bY$ with 
finite second moments it is defined as
\begin{equation*}
  \dCov(\bX,\bY) = \E[ ||\bX-\bX^{\prime}||\,
  ||\bY-\bY^{\prime}||] + \E[||\bX-\bX^{\prime}||]
  \E[||\bY-\bY^{\prime}||] 
	-2 \E[ ||\bX-\bX^{\prime}||\,
  ||\bY-\bY^{\prime\prime}||],
\end{equation*}
where $(\bX',\bY')$ is an independent copy of 
$(\bX,\bY)$, $\bY''$ is another independent copy
of $\bY$, and $||\ . \ ||$ is the Euclidean norm.
The distance covariance is a measure of dependence 
that is always nonnegative, and it is zero if and 
only if the random vectors are independent:
\begin{equation*}
   \dCov(\bX,\bY) = 0 \iff \bX \indep \bY\;,
\end{equation*}
where $\indep$ stands for ``independent of''. 
The arrow $\Leftarrow$ holds for the usual 
covariance as well, but definitely not the arrow 
$\Rightarrow$ that is the main strength of dCov.

\cite{szekely2007dcor} also defined the {\it distance
variance}, {\it distance standard deviation}, and
{\it distance correlation}, given by
\begin{align*}
    \dVar(\bX) &= \dCov(\bX,\bX)\\ 
    \qquad \dStd(\bX) &= \sqrt{\dVar(\bX)}\\
    \dCor(\bX,\bY) &=  
    \frac{\dCov(\bX,\bY)}{\dStd(\bX) \dStd(\bY)}.
\end{align*}
The distance correlation dCor has the nice property 
that it always lies in the interval $[0,1]$. This 
makes it more easily interpretable than dCov, which 
depends on the magnitude and units of the data. 
The empirical distance covariance of multivariate
datasets $(\bX_n,\bY_n)$ is given by
\begin{equation*}
   \dCov_n(\bXn,\bYn) = T_{1,n}(\bXn,\bYn) +
   T_{2,n}(\bXn,\bYn) - T_{3,n}(\bXn,\bYn)
\end{equation*}
with the following terms: 
\begin{align*}
T_{1,n}(\bXn,\bYn) &= \binom{n}{2}^{-1} 
  \sum_{i<j} 
	 ||\bx_i - \bx_j||\,||\by_i - \by_j||\;,\\
T_{2,n}(\bXn,\bYn) &= \left[ \binom{n}{2}^{-1}
  \sum_{i<j} ||\bx_i - \bx_j|| \right] 
	\left[ \binom{n}{2}^{-1} 
	\sum_{i<j} ||\by_i - \by_j||\right],\\
T_{3,n}(\bXn,\bYn) &= \binom{n}{3}^{-1} 
  \sum_{i<j<k}
  \Big[||\bx_i - \bx_j||\,||\by_i - \by_k||
	 + ||\bx_i - \bx_k||\,||\by_i - \by_j||
	 + ||\bx_i - \bx_j||\,||\by_j - \by_k|| \\ 
	&\qquad \qquad \qquad
	 + ||\bx_j - \bx_k||\,||\by_i - \by_j||
	 + ||\bx_i - \bx_k||\,||\by_j - \by_k||
	 + ||\bx_j - \bx_k||\,||\by_i - \by_k||\Big]\,.
\end{align*}
The empirical distance variance and the distance 
correlation are defined similarly as
\begin{align*}
    \dVar_n(\bXn) &= \dCov_n(\bXn,\bXn),\\
    \dCor_n(\bXn,\bYn) &= 
    \frac{\dCov_n(\bXn,\bYn)}
    {\sqrt{\dVar_n(\bXn) \dVar_n(\bYn)}}\;.
\end{align*}

\subsection{dCovICA}
\label{subsec:dcovica}
\cite{dCovICAmatteson2017} introduced dCovICA, an 
ICA method that uses distance covariance to measure 
dependence between sets of potential sources. The
ICA model assumes that the sources 
$\bS = \bA^{-1} \bX$ are mutually independent at
the population level. The goal is to find an 
unmixing matrix that minimizes the empirical 
dependence between the unmixed components. In order
to establish this mutual independence, 
\cite{jin2018generalizing} showed that it suffices
to ensure $d-1$ pairwise independencies. Their
result states that the variables of $\bS$ are 
mutually independent if and only if
\begin{equation*}
    \bS_1 \indep \bS_{1_{+}}\,,\; 
    \bS_2 \indep \bS_{2_{+}}\,,\; 
    \dots, \bS_{d-1} \indep \bS_d\,,
\end{equation*}
where $k_{+} := \{ l: k < l \leqslant d \}$.

The dCovICA method exploits this property by 
constructing its objective function as the sum
of the empirical dependence dCov between each 
of these $d-1$ pairs of a random variable and
a random vector. More precisely, dCovICA starts 
by whitening the observed data $\bX_n$ using 
\begin{equation} \label{covwhitening}
   \bZn = \bXn \bhSigma^{-1/2}_n 
\end{equation}
where $\bhSigma_n$ is the empirical covariance 
matrix. Then the separating matrix $\bU$ is 
estimated by the optimization 
\begin{equation} \label{eq:dcovICA}
    \bUh := \argmin_{\bU} \sum_{k=1}^{d-1}
		\dCov_n(\bSnk(\bU),\bSnkplus(\bU))\,.
\end{equation} 
Here $\bSnk(\bU)$ denotes the $k$-th column of 
$\bSn(\bU) := \bZn\bU^T$, and $\bSnkplus(\bU)$ 
denotes the submatrix of $\bSn(\bU)$ obtained 
by selecting the columns in 
$k_{+} = \{ l: k < l \leqslant d \}$.

In the population setting we have the underlying
model $\bZ = \bUo^T\bS$, and replacing $\bU$ by
$\bUo$ in the right hand side of 
\ref{eq:dcovICA} would zero each term
$\dCov(\bS_k(\bUo),\bS_{k+}(\bUo))$. Then the
objective would become zero, and therefore attain
its minimal value.

\section{Robust independent component analysis}
\label{sec:methodology}

In this section we introduce our proposal for robust 
independent component analysis. Our approach is 
inspired by dCovICA, but we will introduce robustness 
by using a robust objective function and making some
other modifications.

\subsection{Robust distance correlation}
\label{subsec:robin}

The distance covariance and correlation are often 
credited with inherent robustness properties. Such
properties were formally investigated by 
\cite{leyder2024distance} in the setting of 
measuring dependence between univariate random 
variables $X$ and $Y$, with rather nuanced
conclusions. It turned out that while distance 
correlation exhibits some robustness in the sense
of a bounded influence function, it lacks strong 
robustness due to its breakdown value of 0 and
the fact that its influence function does not
go to zero for far outliers. To mitigate its 
sensitivity to outliers, its robustness was
enhanced by a data transformation. 
The univariate datasets are first
standardized to have a median of 0 and a median 
absolute deviation of 1. The data is then 
transformed using the new 
\textit{biloop transform}, and finally the 
distance correlation is computed on the transformed 
data. The biloop transform is given by
$\psi_{\biloop}: \mathbb{R} \to \mathbb{R}^2: 
x \mapsto (u(x),v(x))$ with
\begin{align*}
u(x) &= 
\begin{cases}
   \;\;\;c\, ( 1 + \cos{(2 \pi \tanh{(x/c)}+ \pi)}) 
   & \text{if } x\geqslant 0\\
   -c\, ( 1 + \cos{(2 \pi \tanh{(x/c)}- \pi)})  
   & \text{if } x < 0
\end{cases}\\
v(x) &= \sin(2 \pi \tanh{(x/c)})
\end{align*}
where $c>0$ is a tuning constant which is set to
$c = 4$ by default. Figure \ref{fig:psibiloop} 
illustrates this transformation. The innovative 
feature of $\psi_{\biloop}$ is its two-dimensional 
image, allowing it to go to zero for far outliers,
called the {\it redescending property}, without 
crossing itself. 
This makes it possible for this transformation 
to be simultaneously one-to-one, continuous, and 
redescending. It needs to be one-to-one to preserve 
the independence property, meaning that on the 
population level, the biloop distance correlation 
is zero if and only if the input variables are 
independent. Due to the redescending nature of the 
biloop transform, the biloop distance 
correlation is substantially more robust than the 
classical distance correlation, or the distance 
correlation applied to ranks.

\begin{figure}[!ht]
\centering
\includegraphics[width = 0.50\textwidth]
  {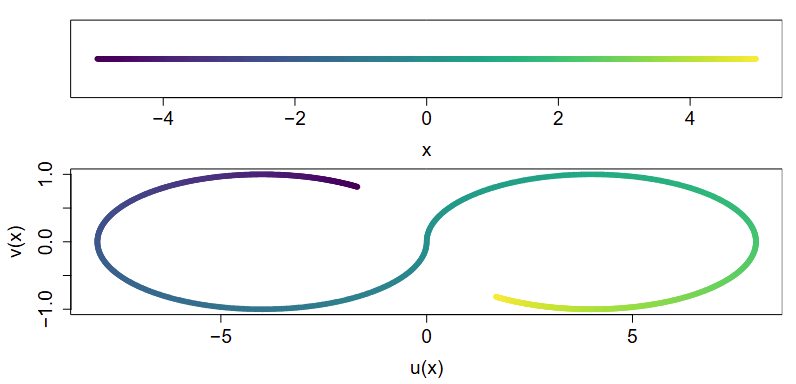}
\caption{Illustration of the biloop transformation.}
\label{fig:psibiloop}
\end{figure}

The biloop transformation was introduced for 
measuring dependence between two univariate random 
variables. In order to obtain a robust dependence
measure between multivariate random variables, one
could apply the biloop transform to each 
coordinate of the vectors separately, and then 
compute the distance correlation between these 
transformed vectors. This is a viable approach,
which transforms a $p$-dimensional vector (with
$p \leqslant d$) such 
that the transformation is one-to-one, continuous, 
and redescending. But the image of the vector 
is now in $2p$-dimensional space, which makes
the computation of subsequent distance 
correlations more computationally intensive. 

This poses an interesting question. Can we 
transform a $p$-dimensional vector such that the 
transformation is bounded, one-to-one, 
continuous, and redescending, but with the 
dimension of its image lower than $2p$\,? The 
answer turns out to be positive,
with a transformation to $(p+1)$-dimensional
space that achieves the three desired properties.
The proposed transformation we call the {\it bowl 
transform}, defined as $\psi: \mathbb{R}^p \to 
\mathbb{R}^{p+1}: \bx \mapsto (v_1(\bx),v_2(\bx))$
in which
\begin{align} \label{eq:bowltransform}
  u(\bx): \mathbb{R}^p \to \mathbb{R}^1: \ 
  &\bx \mapsto \tanh{\left( 
   \frac{||\bx||}{q}\right)} \nonumber \\
  v_1(\bx): \mathbb{R}^p \to \mathbb{R}^p: \ 
  &\bx \mapsto 10 \ u(\bx)^2 \ (1-u(\bx))^2 \ \bx\\
  v_2(\bx): \mathbb{R}^p \to \mathbb{R}^1: \ 
  &\bx \mapsto 10 \ u(\bx)^6 \ 
   (1-u(\bx))^2 \nonumber
\end{align}
where $||\bx||$ is the Euclidean norm of $\bx$
and $q = \sqrt{\chi^2_{0.9975, p}}$ is a 
scaling constant. 
Note that $v_1(\bx)$ is bounded because
$(1 - u(\bx)) \rightarrow 0$ for large $\bx$. 
The transformation is illustrated in 
Figure~\ref{fig:bowl1} for $p = 1$ and in 
Figure~\ref{fig:bowl2} for $p=2$. We see that 
the image of the transformation for $p=2$ 
is a rotated version of that for $p=1$ around 
the $z$-axis. This would not have been possible 
with the biloop transform, because rotating it 
would have broken the continuity.

\begin{figure}[!ht]
\centering
\includegraphics[width=0.50\linewidth]
     {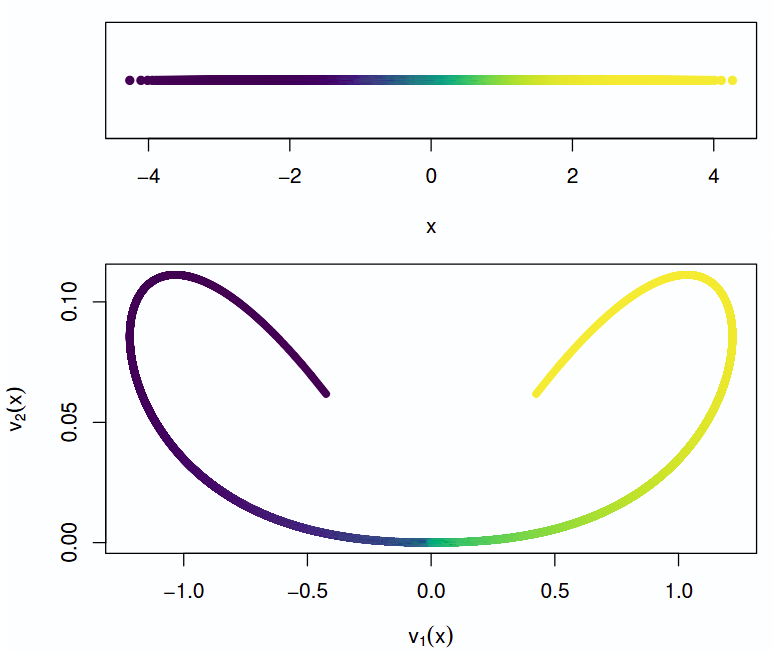}
\caption{Illustration of the bowl transform 
  $\psi(\bx) = (v_1(\bx),v_2(\bx))$ for $p=1$.}
\label{fig:bowl1}
\end{figure}

\begin{figure}[!ht]
\centering
\includegraphics[width=0.55\linewidth]
    {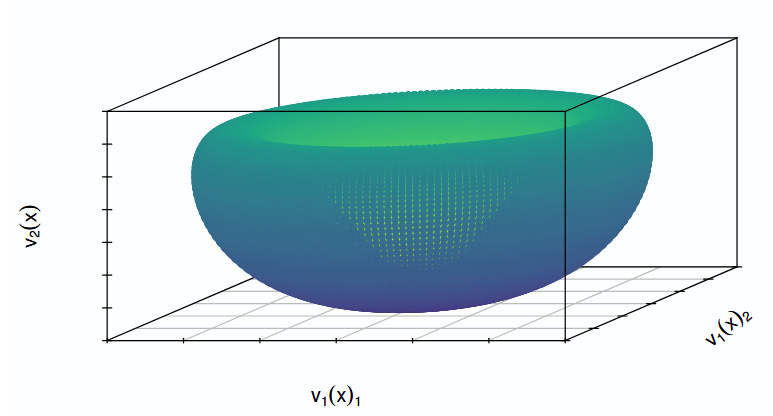}
\caption{Illustration of the bowl transform 
  $\psi(\bx) = (v_1(\bx),v_2(\bx))$ for $p=2$.}
\label{fig:bowl2}
\end{figure}

Computing the distance correlation of random 
vectors $\bX$ and $\bY$ transformed 
by~\eqref{eq:bowltransform} yields 
the more robust dependence measure 
\begin{equation} \label{eq:bowlcor}
\dCor(\psi(\bX),\psi(\bY))\,.
\end{equation}
It still holds that~\eqref{eq:bowlcor} is zero if 
and only if $\bX$ and $\bY$ are independent, and
the dimensions of $\bX$ and $\bY$ only get
increased by 1.

\subsection{RICA}
In order to perform robust independent component 
analysis, we start by whitening the data.  
Classical ICA and dCovICA whiten the data by 
\eqref{covwhitening} which employs the empirical 
covariance matrix $\bhSigma_n$. From a 
robustness perspective this is not ideal, as the 
empirical covariance matrix is highly sensitive
to outliers. Its breakdown value is zero,
meaning that even a single outlier can render the 
estimate useless. Therefore it is essential to 
perform whitening by a robust covariance estimator 
before conducting the next ICA steps, as noted by 
\cite{RADICAL}, \cite{chen2013gammaica}, and 
\cite{nordhausen2018independent}. For this whitening 
step we employ the Minimum Covariance Determinant 
(MCD) estimator of \cite{rousseeuw1984least} 
computed by the algorithm of 
\cite{rousseeuw1999fastMCD}. The MCD has excellent 
robustness properties, due to its bounded influence 
function and its breakdown value of 25\% when using 
the default settings. This process yields the 
robustly whitened data 
\begin{equation} \label{eq:MCDwhitening}
  \bZn = \bXn \bhSigma_{\scriptscriptstyle 
  \text{MCD}}^{-1/2}\;.
\end{equation}

After the data has been whitened, we need to 
estimate an orthogonal separating matrix. The goal
is to find an orthogonal matrix $\bU$ that 
minimizes the empirical dependence between the 
components of $\bSn(\bU) = \bZn \bU^T$. 
For this purpose the dCovICA method employs the
objective function \eqref{eq:dcovICA} based on
the classical dCov measure. We modify this 
objective function in two ways. First, we replace
the distance covariance dCov by the distance
correlation dCor, which is invariant to scale
differences. And secondly, we apply
dCor to the data after subjecting it to the 
bowl transform. We therefore minimize
\begin{equation} \label{eq:RICAobjective}
 \bUh = \argmin_{\bU} \sum_{k=1}^{d-1}\dCor_n
 (\psi(\bSnk(\bU)),\psi(\bSnkplus(\bU)))\;,
\end{equation} 
where again $\bSnk(\bU)$ denotes the $k$-th 
column of $\bSn(\bU)$, and $\bSnkplus(\bU))$ 
denotes the submatrix of $\bSn(\bU)$ consisting
of the columns in $k_{+} = \{ l: k < l \leqslant d \}$.

Combining the estimated separation matrix $\bUh$
from \eqref{eq:RICAobjective} with the whitening 
matrix yields the RICA estimate of the unmixing 
matrix given by
\begin{equation} \label{eq:estRICA}
  \bhA^{\,-1} = \bhSigma_{\scriptscriptstyle
  \text{MCD}}^{-1/2} \bUh \;.
\end{equation}

\subsection{Optimization}
\label{subsec:optimization}

Finding the optimum of the objective function 
in \eqref{eq:RICAobjective} is not an easy task. 
Following \cite{dCovICAmatteson2017} we start by 
parametrizing the estimand $\bU$ as a product of 
$d(d-1)/2$ special matrices. Because $\bU$ is an 
orthogonal matrix, there exists a doubly-indexed 
set of $d(d-1)/2$ angles $\btheta \in \Theta =
\{ \theta_{i,j}: 1 \leqslant i<j\leqslant d : 0 \leqslant 
\theta_{1,j} <2\pi, 0 \leqslant \theta_{i,j} < 
\pi \text{ for } i \neq 1  \}$ such that $\bU$ 
can be written as 
\begin{equation} \label{eq:decomp}
  \bU(\btheta) = \bQ^{d-1} \cdot \bQ^{d-2}
  \cdot \dotso  \cdot \bQ^1\;,
\end{equation}
where
\begin{equation*}
  \bQ^k = \bQ_{k,d}(\theta_{k,d}) \cdot 
  \bQ_{k,d-1}(\theta_{k,d-1}) \cdot \dotso  
  \cdot \bQ_{k,k+1}(\theta_{k,k+1})\,.
\end{equation*}
In the above expressions, 
$\bQ_{i,j}(\theta_{i,j})$ denotes a Givens 
rotation matrix. It is derived from the 
identity matrix $\bI_d$ where the entries 
at $(i,i)$ and $(j,j)$ are replaced by $\cos{(\theta_{i,j})}$, the $(i,j)$ entry by $\sin{(\theta_{i,j})}$, and the 
$(j,i)$ entry by $-\sin{(\theta_{i,j})}$. 
It is worth noting that there exists a unique 
$\btheta$ in $\Theta$ yielding $\bU(\btheta)$
whenever all entries on the diagonal of $\bU$
are positive or all entries of $\bU$ are 
nonzero.

We can thus rewrite the objective function in 
terms of $\btheta$, where we are now 
optimizing over a set of $d(d-1)/2$ angles:
\begin{equation}
\label{eq:RICAobjective_angles}
   \bhtheta = \argmin_{\btheta \in \Theta} 
   \sum_{k=1}^{d-1} \dCor_n
   (\psi(\bSnk(\bU(\btheta))),
   \psi(\bSnkplus(\bU(\btheta))))\,.
\end{equation} 
The decomposition \eqref{eq:decomp} implies
that the $k^{th}$ row of $\bU(\btheta)$ depends 
only on the $k^{th}$ row of a product with 
fewer factors, namely $\bQ^{k} \cdot \bQ^{k-1} 
\cdot \dotso \cdot \bQ^1$. Therefore the 
$k^{th}$ independent component $\bS_k(\btheta)$ 
only depends on the angles $\{ \theta_{i,j}: 
i<j\leqslant d \text{ for } i = 1,2,\dots,k \}$, 
which allows us to estimate the angles 
sequentially. First, estimate the $d - 1$ 
angles $\theta_{1,j}$ yielding the estimate 
$\bS_1(\btheta)$. Next, estimate the $d - 2$ 
angles for recovering $\bS_2(\btheta)$ in the 
subspace orthogonal to the one spanned by 
$\bS_1(\btheta)$. 
We continue this process, each time working in 
the lower-dimensional subspace orthogonal to 
the space spanned by the previously estimated 
components, and estimate the angles of the 
next component. Compared to the joint 
estimation of all the angles 
in \eqref{eq:RICAobjective_angles}, the 
sequential approach is computationally cheaper 
and faster. In principle this could come at 
the cost of a lower accuracy. However, in our 
experiments and using the optimization 
procedure and refinement steps described
below, we found that this was not
the case. Therefore RICA will use the 
sequential approach for all results in this 
paper.

The sequential estimation in RICA gives rise 
to $d-1$ optimization problems, each time 
estimating $d-k$ angles 
$\theta_{k,k+1},\dots,\theta_{k,d}$\,. 
For these optimizations we employ the 
derivative-free solver COBYQA (Constrained 
Optimization BY Quadratic Approximations) of
\cite{COBYQA}. This choice is motivated by 
the challenges of our setting, the 
computation of the bowl transform and the
$\mathcal{O}(n^2)$ computational cost per 
evaluation of the distance correlation,  
which together make it challenging to derive 
and numerically compute exact gradients. As 
a derivative-free, trust-region, sequential 
quadratic programming method, COBYQA is 
designed for such scenarios. It operates by 
iteratively building and refining quadratic 
models of the objective and constraint 
functions using the derivative-free symmetric 
Broyden update, effectively modeling our 
complex landscape. For angle optimization it 
is important that COBYQA always respects 
region constraints, ensuring that all iterated
angles remain within the specified ranges. 
The combination of COBYQA's efficient 
derivative-free optimization and its strict 
adherence to essential bound constraints
makes it a well-suited choice for our 
optimization task.

Finally, we refine our optimization by 
incorporating a technique from the RADICAL 
method of \cite{RADICAL}. They also optimize 
by using $d(d-1)/2$ two-dimensional rotations 
to estimate the orthogonal separating matrix 
$\bU$, a process they refer to as a ``sweep''. 
However, after obtaining a matrix of
sources $\bS_n$ through this process,
they iteratively apply such sweeps multiple 
times until the change in the matrix $\bUh$ 
falls below some threshold. We have opted to 
perform multiple sweeps also. In particular, 
after optimizing the $d(d-1)/2$ angles in 
$\btheta$, we permute the columns of the 
resulting $\bS_n(\btheta)$ and repeat the 
optimization on it, referring to this as 
a sweep. Empirically, we found that 
performing $d+1$ sweeps substantially 
improved the accuracy, in line with the 
findings in the RADICAL paper. On the other 
hand, this also increased the runtime.
The use of sweeps thus involves a trade-off 
between accuracy and computational 
efficiency. Therefore, we will compare RICA 
with the version without sweeps in the 
simulations in Section~\ref{sec:simulations}.

\section{Theoretical analysis}
\label{sec:theoretical analysis}

In this section we establish almost sure 
convergence and root-$n$ consistency of
the RICA estimator. There are two key 
differences between these results and those
previously obtained by 
\cite{dCovICAmatteson2017}. The first is that 
we use the distance correlation rather than 
the distance covariance. The second is that 
we have introduced the bowl transform $\psi$ 
which is applied to the data before measuring 
dependence. As the bowl transform is 
bounded, the first and second moments of the 
transformed data always exist. Therefore, we 
do not require any moment assumptions on the 
original data distribution. Moreover, we will
use the fact that the bowl transform is 
Lipschitz continuous, as it is continuously 
differentiable with bounded derivative.

First we need the population counterpart of  
\eqref{eq:RICAobjective_angles}. We define 
the population parameter $\theta_0$ as
\begin{equation}
\label{eq:RICAobjective_angles_pop}
   \btheta_0 = \argmin_{\btheta \in \Theta}
   \sum_{k=1}^{d-1} \dCor(
   \psi(\bS_k(\bU(\btheta))),
   \psi(\bS_{k_+}(\bU(\btheta))))\,.
\end{equation} 
The following proposition states the almost 
sure convergence of $\bhtheta_n$ to 
$\btheta_0$. We consider a large compact 
subset $\overline{\Theta}$ of the parameter 
space $\Theta$.

\begin{proposition} \label{prop:asconv}
If there exists a unique minimum $\btheta_0$ 
of \eqref{eq:RICAobjective_angles_pop} over
$\overline{\Theta}$, and 
$\btheta_0$ is the only $\btheta$ in $\Theta$
yielding $\bU(\btheta) = \bU(\btheta_0)$,
then we have almost sure convergence for its 
empirical counterpart in 
\eqref{eq:RICAobjective_angles}, that is 
$\bhtheta_n \overset{a.s.}{\to} \btheta_0$.
\end{proposition}

Further details on the required lemmas, their 
proofs, and the proof of the proposition are 
provided in the Supplementary Material. Next, 
we present the result on root-$n$ consistency, 
with its proof also available in the 
Supplementary Material.

\begin{proposition} \label{prop:rootncons}
Under the same assumptions as Proposition 
\ref{prop:asconv} we have root-$n$ 
consistency, that is, 
$||\bhtheta_n - \btheta_0|| = O_P(n^{-1/2})$. 
\end{proposition}

We conclude that if the ICA model 
$\bS = \bU(\btheta_0) \bZ$ holds, the 
estimator in \ref{eq:RICAobjective} is 
strongly consistent and converges at the 
usual parametric rate of $1/\sqrt{n}$.

\section{Simulation results}
\label{sec:simulations}

In this section we evaluate RICA by an extensive 
simulation study. We compare its performance 
against various 
state-of-the-art ICA techniques. We will examine 
different contamination scenarios, sample from a 
range of distributions, and consider different 
dimensions. 

\subsection{Setting}
\label{sec:simsetting}

Simulations evaluating the performance of ICA 
methods typically rely on a standardized set of 
benchmark distributions, as in 
\cite{bach2002kernel}, \cite{RADICAL}, and
\cite{dCovICAmatteson2017}. Following this 
convention, we base our analysis on the 18 
distributions outlined in 
\cite{ElemOfStatLearning}. The corresponding 
distributions, depicted in 
Figure~\ref{fig:18distr}, include Student-$t$ 
distributions, symmetric and asymmetric 
Gaussian mixtures, and exponential mixtures.

\begin{figure}[!ht]
\centering
\includegraphics[width=0.6\linewidth]
   {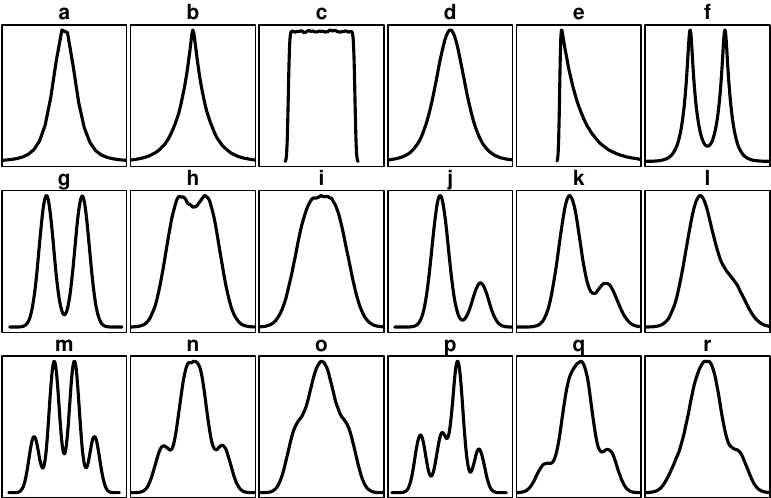}
\caption{The 18 source distributions used 
    in the simulations.}
\label{fig:18distr}
\end{figure}

We consider 3 main settings for sampling the 
data, depending on the contamination added. 
The first setting comprises uncontaminated data. 
For this, the columns of the $n \times d$ matrix 
of data sources $\bSn$ are sampled from the 18 
distributions. 
For the contaminated settings we consider two 
scenarios from \cite{brys2005robustification}, 
called multiplicative and clustered 
contamination. 

For clustered contamination we begin 
with an uncontaminated $\bSn$ as before, 
but then replace 10\% of the rows by 
observations sampled from $\mathcal{N}(15,1)$. 
For multiplicative contamination, we
instead replace 10\% of the rows as follows.
For each selected row, we generate a vector of
size $d$ with values sampled from the uniform 
distribution on $[0,1]$, ensuring that its sum 
is greater than 1. Then this vector is 
multiplied by $d$ and by either the minimum or 
maximum value of the source in each column, 
with the choice between the two made randomly, 
before replacing the row with it. In this way 
outliers are positioned on either side of the 
bulk of the data points. 

After the independent sources $\bSn$ are 
generated, we generate random 
$d \times d$ mixing matrices $A$ 
that are well-conditioned in the sense that 
their condition number is between 1 and 2. 
We then compute the mixed data 
$\bXn = \bSn \bA^T$ that is the
input of the ICA methods. 

We sample data in this manner for three 
different dimensions $d$ as follows:
\begin{itemize}
\item For $d=2$ we generate $n=1000$ 
  observations per dataset and perform 1000 
  replications;
\item For $d=4$ we generate $n=1000$ 
  observations per dataset and perform 100 
  replications;
\item For $d=6$ we generate $n=2000$ 
  observations per dataset and perform 100 
  replications.
\end{itemize}
We repeat this process for each of the 18 
distributions and for each of the three 
contamination settings.

We compare RICA to the FastICA, JADE, 
infomax, and dCovICA methods, that start 
with their classical whitening step. We 
also include the robust RADICAL method, 
that we run with its built-in robust 
whitening step.

For FastICA we used its implementation in 
the Python library \textit{scikit-learn} 
\citep{scikit-learn}, with the exponential 
function as its non-linearity because this 
has the best robustness properties 
\citep{hyvarinen1999robust}. Infomax and 
JADE were obtained from the R package 
\textit{ica} \citep{icapackageR}, while 
RADICAL was sourced from the R package 
\textit{rmgarch} \citep{rmgarchpackageR}. 
For dCovICA and RICA we wrote our own 
Python implementation, using \textit{scipy} 
\citep{SciPy} for the optimization. For
the MCD whitening we used 
\textit{robpy} \citep{leyder2024robpy}.

To examine the performance of the different 
ICA methods we require an evaluation metric. 
\mbox{Direct} comparison of the estimated 
mixing matrices $\bA$ is not possible due to 
inherent indeterminacies in scale, sign, 
and permutation, as discussed in 
Section~\ref{subsec:whitening}. To bypass 
this, we use the Amari error as a measure
of inaccuracy \citep{amarierror1995}:
\begin{equation*}
  \frac{\sum_{i=1}^d \sum_{j=1}^d \left( 
  \frac{|\bP_{ij}|}{\max_k|\bP_{ik}|}-1 
  \right) + \sum_{j=1}^d \sum_{i=1}^d 
  \left( \frac{|\bP_{ij}|}
  {\max_k|\bP_{kj}|}-1 \right)}{2d(d-1)}.
\end{equation*}
Here $\bP = \bU \bW \bA$ is the product of 
the estimated unmixing matrix $\bU \bW$ 
obtained from the ICA algorithm, and the 
true mixing matrix $\bA$. The Amari measure 
is invariant to ICA indeterminacies and 
ranges from 0 to 1, where 0 indicates the 
best possible performance and 1 the worst.

\subsection{Results}
\label{sec:simulationresults}

The simulation results are organized by the 
type of contamination: first those for 
uncontaminated data are shown, then with 
clustered contamination, and finally with 
multiplicative contamination.

The Amari error is averaged over the 
replications per setting, and the average Amari 
error of a method is displayed in function of 
the distribution type on the horizontal axis, 
as in Figure~\ref{fig:sim_uncontaminated}.
This is common practice in the ICA literature
\citep{ElemOfStatLearning}. 
This visualization makes it easier to compare 
ICA methods than by looking at large tables. 
But for completeness, we also provide the 
numerical results in tabular form in 
the Supplementary Material.

\begin{figure}[!ht]
\centering
\begin{subfigure}{0.45\textwidth}
  \centering
  \caption{$d=2$}
  \includegraphics[width=\textwidth]
	{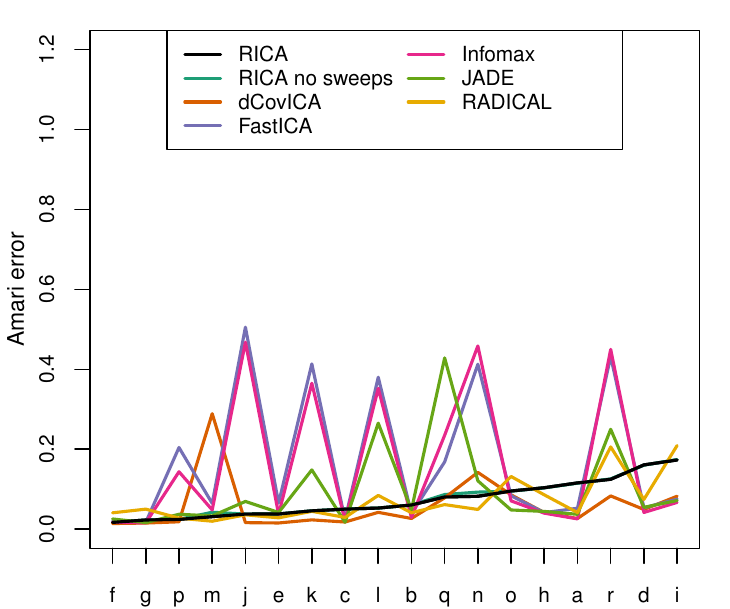}
\end{subfigure}
\begin{subfigure}{0.45\textwidth}
  \centering
  \caption{$d=4$}
  \includegraphics[width=\textwidth]
	{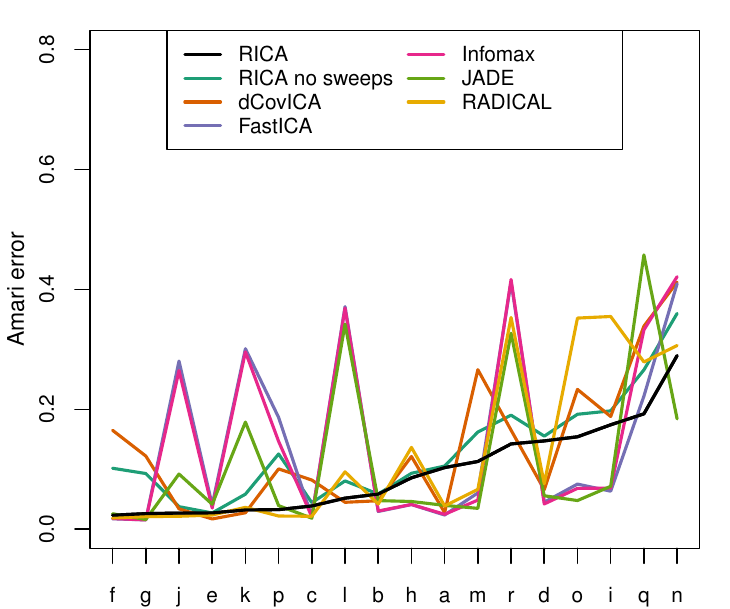}
\end{subfigure}
\begin{subfigure}{0.45\textwidth}
  \centering
  \caption{$d=6$}
  \includegraphics[width=\textwidth]
	{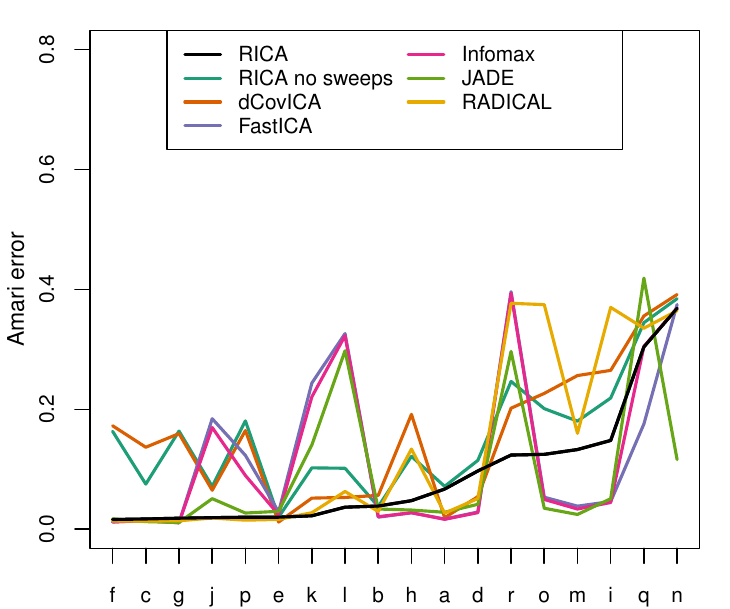}
\end{subfigure}
\begin{subfigure}{0.45\textwidth}
  \centering
  \captionsetup{justification=centering}
  \caption{Amari error ($\times 100$) 
  averaged over\\ all 18 distributions.}
  \hspace{.5cm}
  \raisebox{3.5cm}{
  \scalebox{0.85}{\begin{tabular}{|c|ccc|}
    \hline
    mean & $d=2$ & $d=4$ & $d=6$\\ 
    \hline
    RICA & 7.32 & 
       \textbf{9.54} & \textbf{9.01}\\ 
    RICA no sweeps & 7.47 & 13.04 & 15.55 \\ 
    dCovICA & \textbf{5.93} & 13.65 & 15.74 \\ 
    FastICA & 16.96 & 14.53 & 11.78 \\ 
    Infomax & 16.03 & 14.77 & 11.95 \\ 
    JADE & 9.77 & 11.47 & 9.24 \\ 
    RADICAL & 7.01 & 12.58 & 13.35 \\ 
    \hline
    \end{tabular}}}
\end{subfigure}
\caption{Simulation results for 
         uncontaminated data.}
\label{fig:sim_uncontaminated}
\end{figure}

To enhance the readability and to highlight 
the performance of RICA we have ordered the 
distributions on the horizontal axis of each 
plot according to the results obtained by 
the default version of RICA, which is the 
one with optimization sweeps.

\subsubsection{Uncontaminated data}

Figure \ref{fig:sim_uncontaminated} shows the
simulation results on uncontaminated data, with
one plot per dimension. It also includes a small 
table listing the average Amari error over all 18 
distributions, with the lowest error per column
in boldface.

From the results we conclude that even when 
there is no contamination, RICA performed 
very well. While other methods attain
both high and low errors, reflecting
sensitivity to the underlying distribution 
type, RICA exhibits a more stable performance 
with less variation. In dimensions 4 
and 6 it was the best-performing method 
based on the average error over all 
distributions.

\subsubsection{Clustered contamination}

When 10\% of clustered contamination 
was introduced into the data, 
Figure~\ref{fig:sim_clustcontaminated} 
indicates that most ICA methods lost 
their competitiveness, while RICA 
remained the only method maintaining low 
Amari errors. It is worth noting that a 
random non-singular matrix will typically 
yield an Amari error around 0.4 
\citep{nordhausen2011performanceindices}, 
suggesting that for such clustered 
contamination data, most ICA methods 
performed no better than making a random 
guess.

We also observe that the beneficial 
effect of performing optimization sweeps 
on the accuracy remains as strong as it 
was for uncontaminated data, particularly 
as the dimension increased. This is why 
our default version of RICA is with the
sweeps.

\begin{figure}[!ht]
    \centering
    \begin{subfigure}{0.45\textwidth}
        \centering
        \caption{$d=2$}
        \includegraphics[width=\textwidth]
				{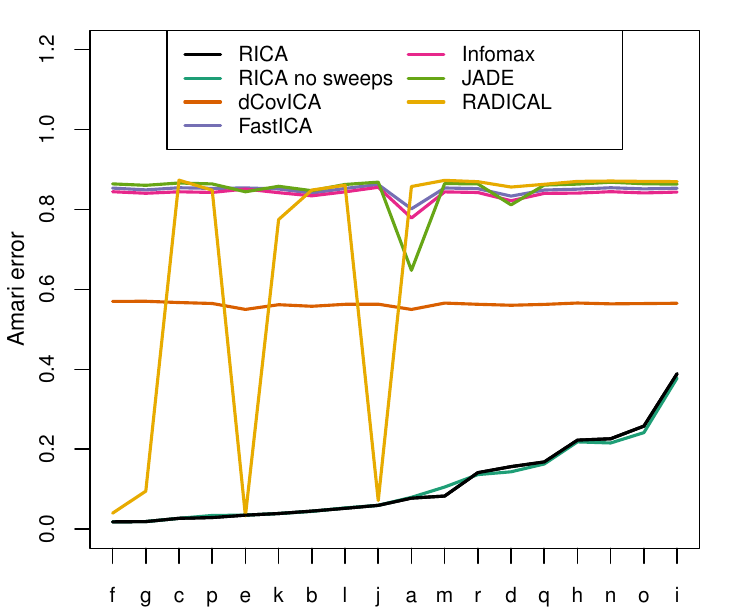}
            \end{subfigure}
    \begin{subfigure}{0.45\textwidth}
        \centering
        \caption{$d=4$}
        \includegraphics[width=\textwidth]
				{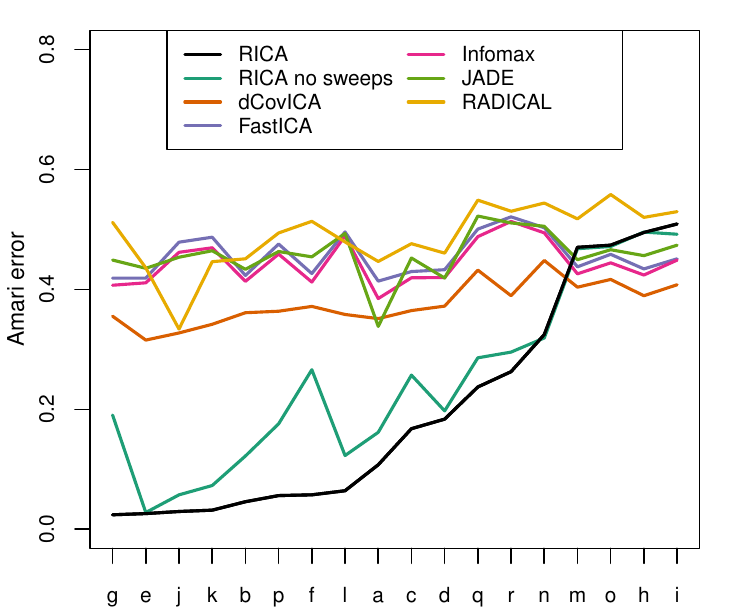}
    \end{subfigure}
    \begin{subfigure}{0.45\textwidth}
        \centering
        \caption{$d=6$}
        \includegraphics[width=\textwidth]
				{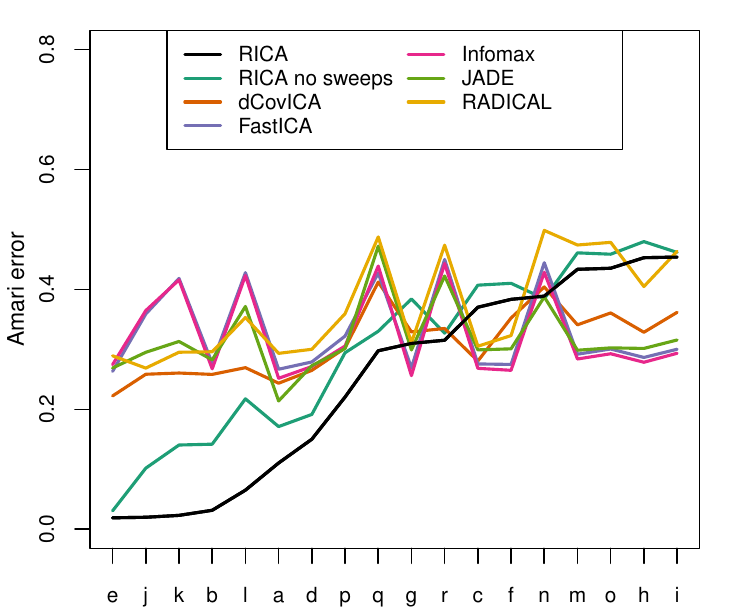}
    \end{subfigure}
    \begin{subfigure}{0.45\textwidth}
        \centering
        \captionsetup{justification=centering}
				\caption{Amari error ($\times 100$) 
        averaged over\\ all 18 distributions.}
        \hspace{.5cm}
        \raisebox{3.5cm}{
        \scalebox{0.85}{\begin{tabular}{|c|ccc|}
          \hline
 & $d=2$ & $d=4$ & $d=6$ \\ 
  \hline
  RICA & 11.37 &
       \textbf{19.81} & \textbf{24.89} \\ 
  RICA no sweeps & \textbf{11.18} & 
       24.88 & 29.96 \\ 
  dCovICA & 56.27 & 37.61 & 31.03 \\ 
  FastICA & 84.85 & 45.60 & 32.97 \\ 
  Infomax & 83.88 & 44.37 & 32.35 \\ 
  JADE & 84.70 & 45.79 & 31.77 \\ 
  RADICAL & 68.04 & 48.89 & 37.09 \\ 
   \hline
        \end{tabular}}}
    \end{subfigure}
    \caption{Simulation results for
       clustered contamination.}
    \label{fig:sim_clustcontaminated}
\end{figure}

\subsubsection{Multiplicative contamination}

Lastly, the results for multiplicative 
contamination are shown in 
Figure~\ref{fig:sim_multcontaminated}. 
While RICA is again the most effective 
method in the presence of outliers, 
RADICAL does outperform the remaining
methods. This is not surprising, as 
insensitivity to outliers was one of
its design goals. However, its 
effectiveness depends on the 
distribution of the sources, as it
worked very well in some settings 
and did worse than a random guess 
in others. We also see that dCovICA 
performs well at certain distributions. 
Indeed, in Section \ref{subsec:robin} 
we saw that the standard dCov has some 
robustness properties, but perhaps 
not enough.

\begin{figure}[!ht]
\centering
    \begin{subfigure}{0.45\textwidth}
        \centering
        \caption{$d=2$}
        \includegraphics[width=\textwidth]
				{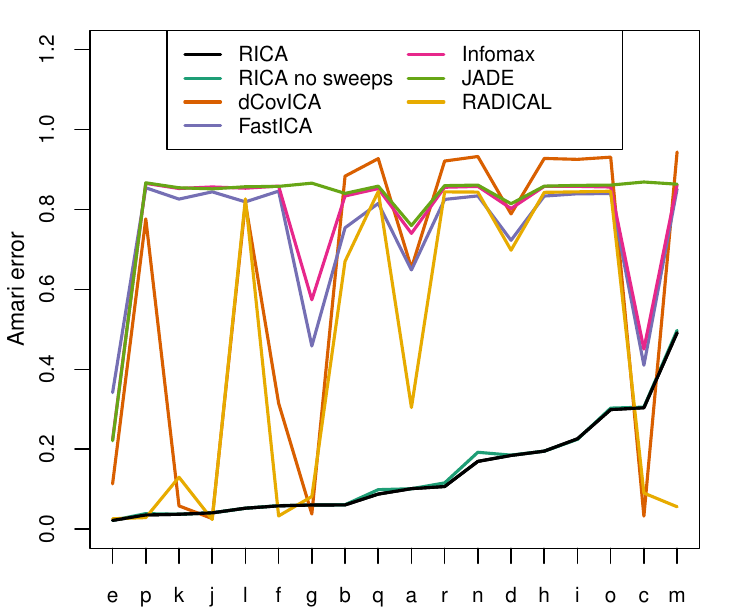}
            \end{subfigure}
    \begin{subfigure}{0.45\textwidth}
        \centering
        \caption{$d=4$}
        \includegraphics[width=\textwidth]
				{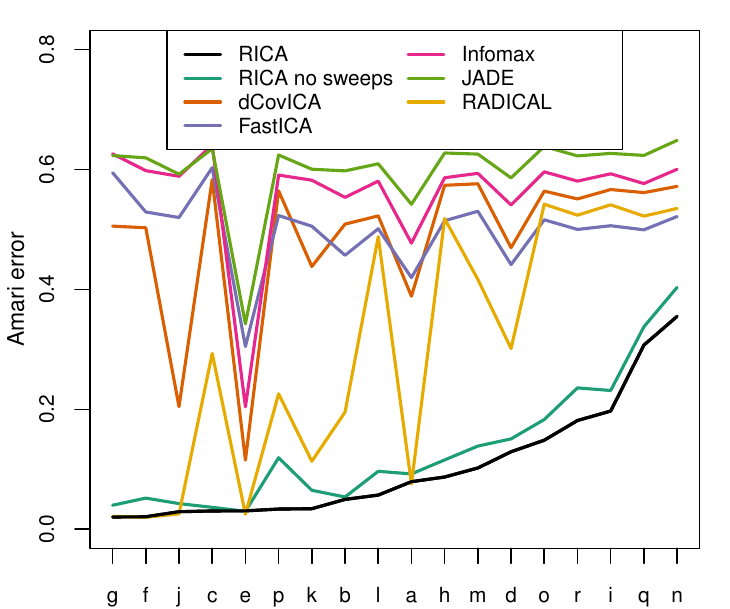}
    \end{subfigure}
    \begin{subfigure}{0.45\textwidth}
        \centering
        \caption{$d=6$}
        \includegraphics[width=\textwidth]
				{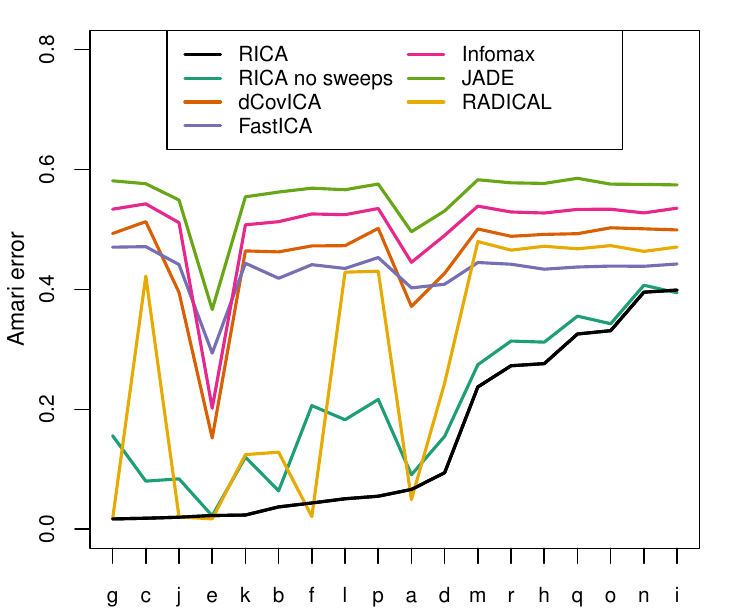}
    \end{subfigure}
    \begin{subfigure}{0.45\textwidth}
        \centering
        \captionsetup{justification=centering}
        \caption{Amari error ($\times 100$) 
        averaged over\\ all 18 distributions.}
        \hspace{.5cm}
        \raisebox{3.5cm}{
        \scalebox{0.85}{\begin{tabular}{|c|ccc|}
          \hline
 & $d=2$ & $d=4$ & $d=6$ \\ 
  \hline
  RICA & \textbf{14.08} & 
      \textbf{10.51} & \textbf{14.92} \\ 
  RICA no sweeps & 14.39 & 13.46 & 20.98 \\ 
  dCovICA & 61.18 & 48.72 & 45.58 \\ 
  FastICA & 74.25 & 49.94 & 43.10 \\ 
  Infomax & 77.29 & 56.17 & 50.32 \\ 
  JADE & 81.56 & 59.93 & 55.43 \\ 
  RADICAL & 44.64 & 29.91 & 28.86 \\ 
   \hline
        \end{tabular}}}
    \end{subfigure}
    \caption{Simulation results for 
      multiplicative contamination.}
    \label{fig:sim_multcontaminated}
\end{figure}

\subsection{Increasing contamination}

Following \cite{bach2002kernel} and 
\cite{RADICAL}, we conducted an experiment 
to investigate the effect of gradually 
introducing more and more outliers. 
Starting from outlier-free datasets, we 
progressively increase the proportion of 
outliers up to 20\%. 

We start with bivariate datasets of 
1000 observations, where each source was
sampled uniformly at random from the 18 
previously defined distributions. We then
introduced outliers step by step,
each time drawing one more data point 
and replacing a randomy selected 
coordinate of that point by either +5 
or -5, with equal probability. 
This is repeated until the number of 
outliers reaches 200, so the 
contamination level ranges from 0\% to 
20\%. The experiment was replicated
1000 times to ensure reliable results. 

\begin{figure}[!ht]
\centering
\includegraphics[width=0.49\textwidth]
  {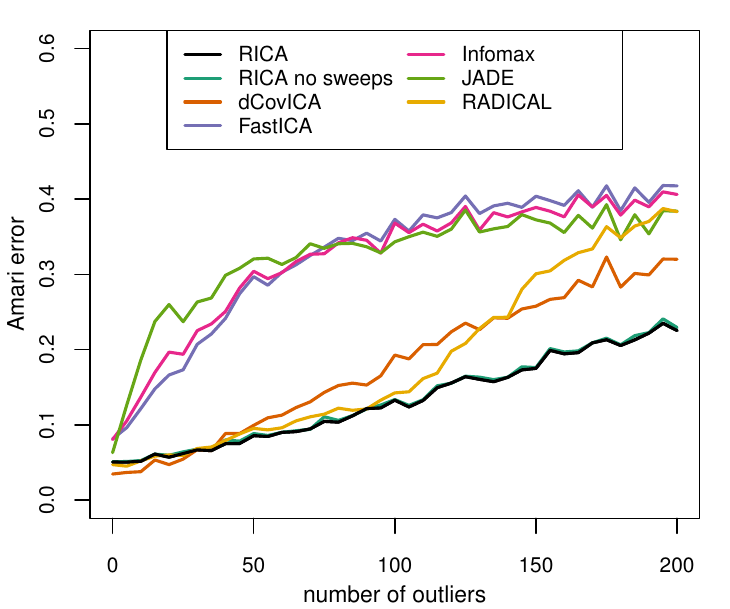}
\caption{Effect of increasing number of 
         outliers for $d=2$.}
\label{fig:increasingcontamination2D}
\end{figure}

Figure~\ref{fig:increasingcontamination2D}
plots the average Amari error in function
of the number of outliers. For all methods
the error goes up with the amount of
contamination, with RICA being the least
affected. RADICAL also achieves good results
initially, but for more than 10\% of outliers
its performance deteriorates. The dCovICA 
method performs well up to 5\% of outliers, 
after which it becomes less effective. 
FastICA, infomax and JADE are affected 
already for a small number of outliers.

\section{Real data examples}
\label{sec:real data examples}

In this section, we illustrate RICA on real 
datasets containing outliers. For this we use 
image data, audio data, and periodic signals.

\subsection{Image data}

As a first example, we use three images from 
the \textit{scikit-image} package in Python 
\citep{scikit-image}: an astronaut, grass, and 
a cat. The grayscale images are $128 \times 
128$. Storing the pixels in a tall column
vector, this yields a $16384 \times 3$ dataset. 
We then introduce outliers in each image by 
randomly selecting 250 pixels and turning them
into black pixels, which amounts to 1.5\% of
contamination. The images are then mixed using 
a random $3 \times 3$ orthogonal matrix $\bA$, 
and we attempt to reconstruct the original ones.

In Figure \ref{Fig:Imagesexample} we see that
RICA is the only method that successfully 
recovers the original images. In contrast, 
the other three methods fail to separate 
the sources correctly, as the cat appears in 
two of the reconstructed images.

\begin{figure}[!ht]
\centering
    \begin{subfigure}{0.45\textwidth}
        \centering
        \caption{original images}
        \includegraphics[width=\textwidth]
				{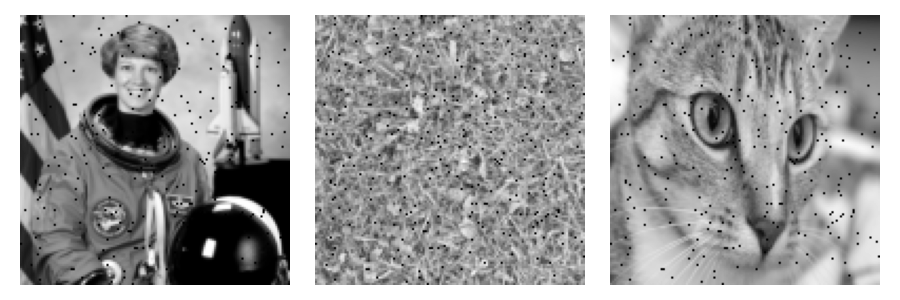}
    \end{subfigure}
    \begin{subfigure}{0.45\textwidth}
        \centering
        \caption{mixed images}
        \includegraphics[width=\textwidth]
				{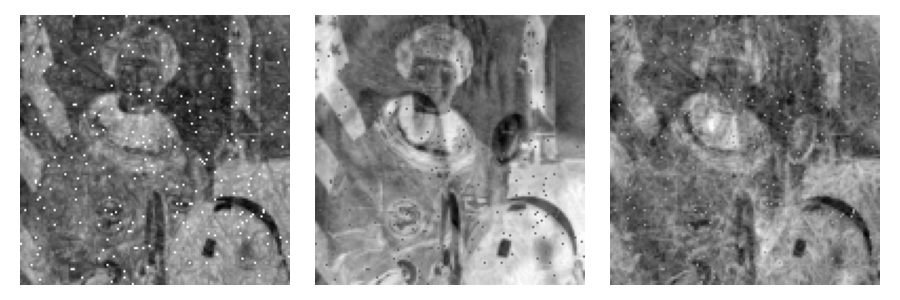}
    \end{subfigure}
    \begin{subfigure}{0.45\textwidth}
        \centering
        \caption{RICA}
        \includegraphics[width=\textwidth]
				{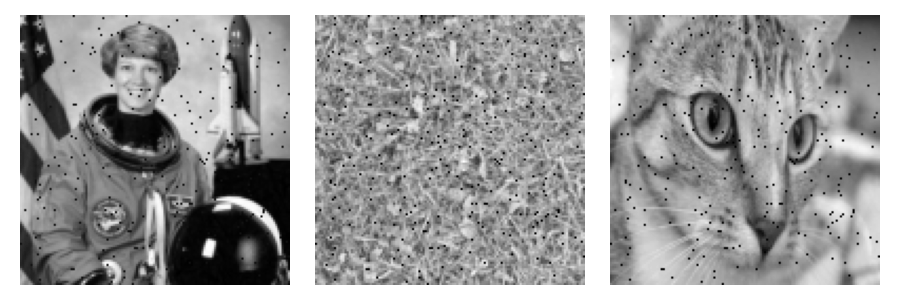}
    \end{subfigure}
    \begin{subfigure}{0.45\textwidth}
        \centering
        \caption{dCovICA}
        \includegraphics[width=\textwidth]
				{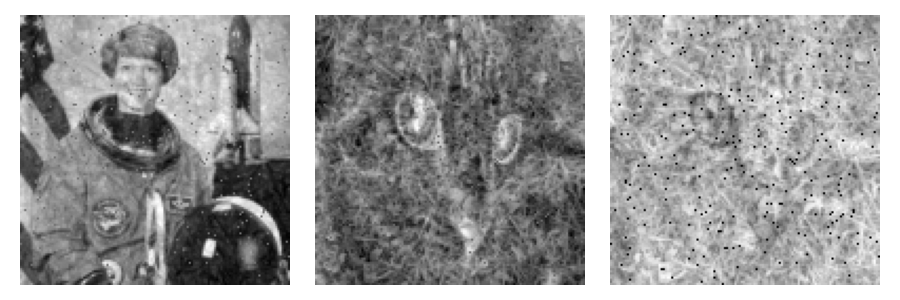}
    \end{subfigure}
    \begin{subfigure}{0.45\textwidth}
        \centering
        \caption{FastICA}
        \includegraphics[width=\textwidth]
				{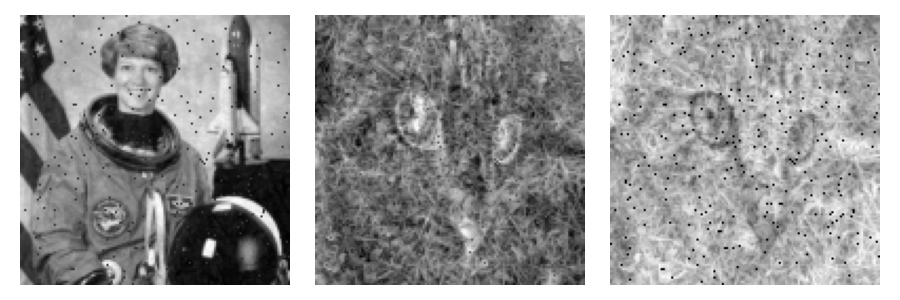}
    \end{subfigure}
    \begin{subfigure}{0.45\textwidth}
        \centering
        \caption{RADICAL}
        \includegraphics[width=\textwidth]
				{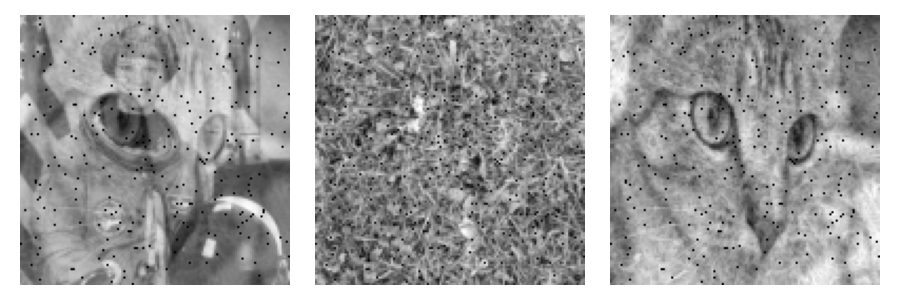}
    \end{subfigure}
\caption{Mixed images example.}
\label{Fig:Imagesexample}
\end{figure}

\subsection{Cocktail party problem}

Another application of ICA is to the 
well-known cocktail party problem, a scenario 
in which multiple overlapping signals, such as 
speech from different speakers, are recorded 
through multiple microphones. The objective is 
to separate these mixed signals without prior 
knowledge about the original sources.

\begin{figure}[!ht]
    \centering
    \begin{subfigure}{0.45\textwidth}
        \centering
        \caption{original signals}
        \includegraphics[width=\textwidth]
				{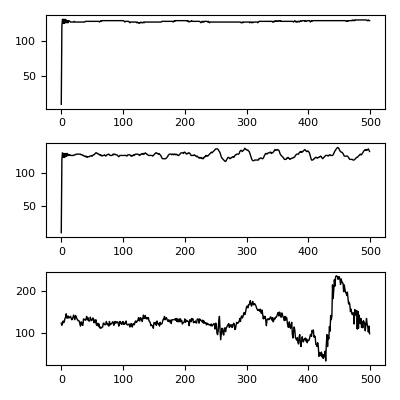}
            \end{subfigure}
    \begin{subfigure}{0.45\textwidth}
        \centering
        \caption{RICA}
        \includegraphics[width=\textwidth]
				{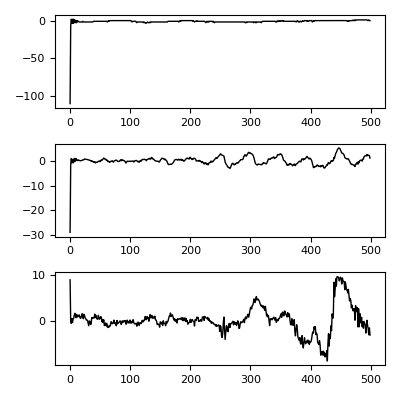}
    \end{subfigure}
    \begin{subfigure}{0.45\textwidth}
        \centering
        \caption{dCovICA}
        \includegraphics[width=\textwidth]
				{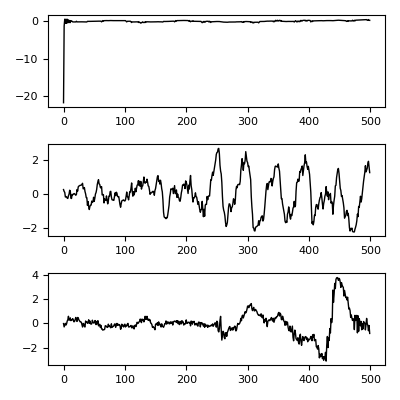}
            \end{subfigure}
    \begin{subfigure}{0.45\textwidth}
        \centering
        \caption{fastICA}
        \includegraphics[width=\textwidth]
				{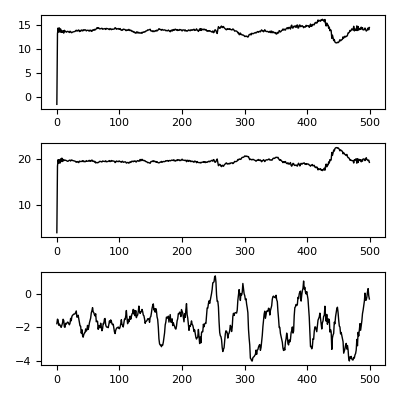}
    \end{subfigure}
    \caption{Cocktail party example on 
		    \textit{JADE} data.}
    \label{Fig:CPPexampleJADE}
\end{figure}

We first illustrate this problem on the 
cocktail party data from the \textit{JADE} 
package \citep{JADEpackage} in R. 
We select the first 500 cases, where an 
artifact stands out in two of the three sources 
at the start of the recording, likely due to 
turning on the microphones or the recorder.
The artifact is clearly visible in the top two
time series in the northwest panel of
Figure~\ref{Fig:CPPexampleJADE}. The three
sources are mixed by a random $3 \times 3$ 
orthogonal matrix.

We then apply RICA, dCovICA, and FastICA. The
results are shown in the three remaining panels
of~Figure \ref{Fig:CPPexampleJADE}. We see that
only RICA successfully separates the signals, 
while both dCovICA and FastICA are significantly 
affected by the outlier. The corresponding Amari 
errors further support this: RICA 
achieves the lowest error (0.0875), followed by 
dCovICA (0.2326) and FastICA (0.4895).

After removing the outlying observation from the 
data set, dCovICA and FastICA perform better, 
with Amari errors of 0.0905 and 0.2714. The error 
of RICA stays about the same at 0.0992, so it 
was minimally affected by the outlier.

A second cocktail party dataset containing outliers 
is from \cite{brys2005robustification}. This dataset 
consists of 5000 observations of two signals, with a
loud external noise occurring in the middle. This 
noise affects both sources, introducing dependence 
between the signals.

\begin{figure}[!ht]
\centering
    \begin{subfigure}{0.45\textwidth}
        \centering
        \caption{original signals}
        \includegraphics[width=\textwidth]
        {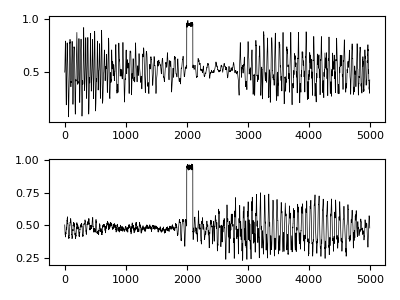}
            \end{subfigure}
    \begin{subfigure}{0.45\textwidth}
        \centering
        \caption{RICA}
        \includegraphics[width=\textwidth]
        {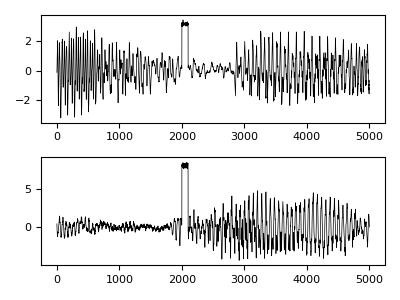}
    \end{subfigure}
    \begin{subfigure}{0.45\textwidth}
        \centering
        \caption{dCovICA}
        \includegraphics[width=\textwidth]
        {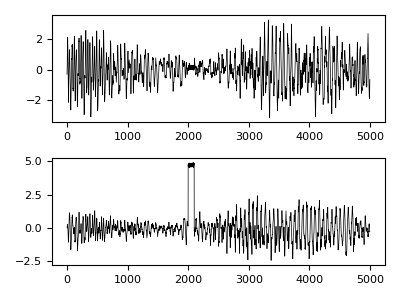}
            \end{subfigure}
    \begin{subfigure}{0.45\textwidth}
        \centering
        \caption{FastICA}
        \includegraphics[width=\textwidth]
        {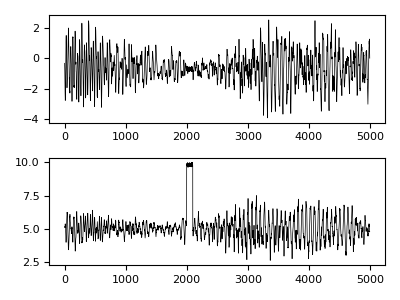}
    \end{subfigure}
\caption{Cocktail party example on the
         Brys et al data.}
\label{Fig:CPPexampleBrys}
\end{figure}

Figure~\ref{Fig:CPPexampleBrys} shows the sources
before mixing, and the unmixing results of RICA,
dCovICA and FastICA. Again RICA is the only method
recovering the original signals, while dCovICA and 
FastICA are affected by the external noise and  
attribute the noise to a single source. RICA  
attains a low Amari error (0.017), while dCovICA 
and FastICA obtain high errors (0.549 and 0.590).

\subsection{Periodic data}

Lastly we replicate the periodic signals from 
\cite{brys2005robustification}. This setup 
involves generating three periodic sources:
\begin{equation*}
    \sin{(3t)}, \ \text{sawtooth}(19t), \ 
    \log (\text{rem}(t,\pi)),
\end{equation*}
for $t=-10\pi,-10\pi+0.1,\dots, 10\pi$. Here 
\texttt{sawtooth} repeats a linear function,
and \texttt{rem} is the remainder function.
We then contaminate the sawtooth signal by
1.5\% of outliers, represented as vertical 
lines in Figure~\ref{Fig:Periodicexample}. 
After mixing the data we apply RICA, dCovICA, 
and FastICA, yielding the remaining panels of
Figure~\ref{Fig:Periodicexample}. Both dCovICA 
and FastICA struggle to recover the third signal, 
which becomes quite noisy and gets an upward
spike near the end. RICA does reproduce the 
signals and attains the lowest Amari error 
(0.0111), followed by dCovICA (0.0742) and 
FastICA (0.1001).

\begin{figure}[!ht]
    \centering
    \begin{subfigure}{0.45\textwidth}
        \centering
        \caption{original signals}
        \includegraphics[width=\textwidth]
				{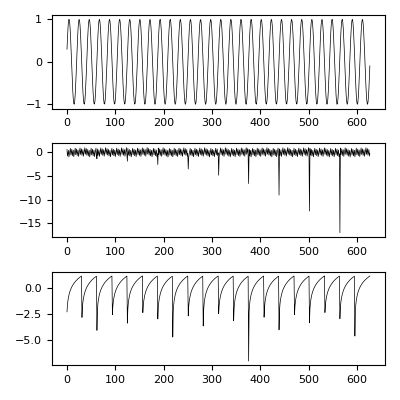}
    \end{subfigure}
    \begin{subfigure}{0.45\textwidth}
        \centering
        \caption{RICA}
        \includegraphics[width=\textwidth]
				{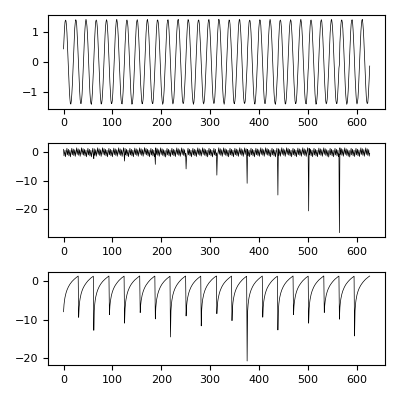}
    \end{subfigure}
    \begin{subfigure}{0.45\textwidth}
        \centering
        \caption{dCovICA}
        \includegraphics[width=\textwidth]
				{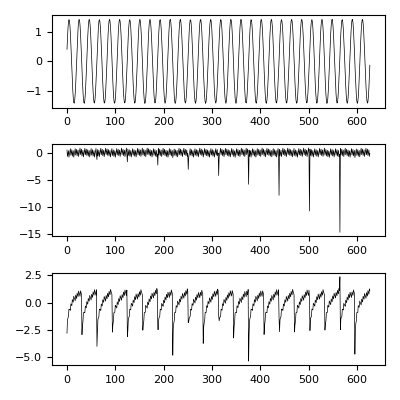}
        \end{subfigure}
    \begin{subfigure}{0.45\textwidth}
        \centering
        \caption{FastICA}
        \includegraphics[width=\textwidth]
				{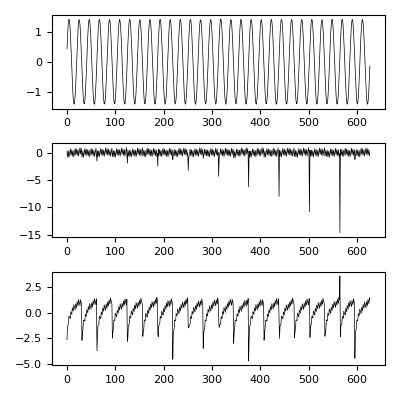}
    \end{subfigure}
    \caption{Periodic data example.}
    \label{Fig:Periodicexample}
\end{figure}

\section{Discussion}
\label{sec:discussion}

Independent component analysis (ICA) is a powerful 
tool in signal processing. Unfortunately, most 
popular approaches for ICA are not robust against 
outliers, for instance because they use 
higher-order statistics. 
In the literature several methods have
been proposed that aim to be more robust, based
on various principles. Some approaches involved
a robust choice of contrast function, and 
others used preprocessing, nonparametrics, or 
divergences. In this paper we proposed a robust 
ICA method called RICA, which proceeds by 
estimating the separating matrix by minimizing
a robust measure of dependence between vector
variables.

The dependence measure used is the distance
correlation (dCor) of \cite{szekely2007dcor}, 
whereas other authors have used the
distance covariance. In order to make the 
distance correlation more robust we constructed
a new transformation called the bowl transform,
which maps $p$-dimensional space to a manifold 
in $p+1$ dimensions. This transform is bounded,
one-to-one, and continuous. It is also
redescending, meaning that far outliers are
mapped to points close to the origin.
By computing the dCor after transforming its 
inputs it becomes more robust to outliers, 
while preserving the crucial property that 
a dCor of zero characterizes independence.
The RICA method estimates the independent
sources sequentially, by looking for the
component which has the smallest dCor with
the remainder projected on its orthogonal
complement.

We showed that RICA is strongly consistent and 
has the usual parametric rate of convergence.
Its robustness was investigated empirically by
a simulation study, in which it outperformed
its competitors most of the time and on average.
RICA was illustrated in three applications,
including the well-known cocktail party problem.

It is worth noting that combining robustness 
with measuring independence is a balancing
act. Robust methods aim to reduce the effect of 
extreme observations. However, if the true 
dependence mainly comes from the tails of the 
distributions of the sources, a robust method
may miss it.
With this in mind, it is recommended to compare 
the results of a robust method with those of a 
classical method. If the outcomes are similar, 
all is well. If the outcomes are different, one 
has to study the results in order to decide
whether the extreme points in question are part 
of the structure we want to model, or due to 
unwanted data contamination.

Among potential directions for future research 
could be extensions of RICA to blind source 
separation, or to estimation in the noisy ICA 
model described in e.g. \citep{voss2015}.\\

\noindent \textbf{Acknowledgment.} We are
grateful to Art Owen for asking the question that 
motivated us to construct the bowl transform.



\clearpage
\pagenumbering{arabic}
\setcounter{page}{1}
\appendix
\begin{center}
\phantom{abc}\\ 

\Large{\bf Supplementary Material}\\
\end{center}
\vspace{15mm}

\setcounter{equation}{0} 
\renewcommand{\theequation}
  {A.\arabic{equation}} 

\spacingset{1.45} 

\section{Proofs of the propositions in 
Section~\ref{sec:theoretical analysis}}

\subsection{Proof of Proposition \ref{prop:asconv}}
To prove this proposition, we first generalize Lemmas 
A.1-A.3 of \cite{dCovICAmatteson2017} to the lemmas
below.
 
\begin{lemma} \label{lemma1}
$\forall \btheta \in \Theta:$
\begin{equation*}
   \sum_{k=1}^{d-1} 
   \dCor_n(\psi(\bSnk(\btheta)),
   \psi(\bSnkplus(\btheta))) 
   \;\; \overset{a.s.}{\rightarrow} \;\; 
   \sum_{k=1}^{d-1} 
   \dCor(\psi(\bS_{ k}(\btheta)),
   \psi(\bS_{ k_{+}}(\btheta))) \ 
   \text{ for } \ n \to \infty.
\end{equation*} 
\end{lemma}
\begin{proof}
Note that
\begin{align*}
   &\left|\sum_{k=1}^{d-1} 
   \dCor_n(\psi(\bSnk(\btheta)),
   \psi(\bSnkplus(\btheta))) - 
   \sum_{k=1}^{d-1} \dCor(\psi(\bS_{k}(\btheta)),
   \psi(\bS_{ k_{+}}(\btheta)))  \right|\\
   &\leqslant \sum_{k=1}^{d-1} 
   \left|\dCor_n(\psi(\bSnk(\btheta)),
   \psi(\bSnkplus(\btheta))) - \dCor(\psi(\bS_{k}
   (\btheta)),\psi(\bS_{k_{+}}(\btheta))) \right|\,.
\end{align*}
Here we know that $\psi(\bS_{ k}(\btheta))$
and $\psi(\bS_{ k_{+}}(\btheta))$ have finite 
moments $\forall k \in \{1,\dots,d-1\}$ as 
$\psi$ is bounded. Hence we can use 
Corollary 1 of \cite{szekely2007dcor} which states 
that $\dCor_n(\bXn,\bYn)$ converges almost surely 
to $\dCor(\bX,\bY)$, establishing the lemma.
\end{proof}

Next, we consider the group of all $d \times d$ 
orthogonal matrices with \mbox{determinant} 1, 
which is denoted by $\mathcal{SO}(d)$. On this 
group, take a pseudometric $\mathcal{D}$ such 
that $\forall \bU, \bA \in \mathcal{SO}(d): 
\mathcal{D}(\bU,\bA) \geq 0$, and 
$\mathcal{D}(\bU,\bA)
= 0 \iff \exists \boldsymbol{P}$, a signed 
permutation matrix, such that 
$\bU = \boldsymbol{P} \bA$. This yields 
equivalence classes on $\mathcal{SO}(d)$ where 
two elements are in the same equivalence class 
if and only if their distance is zero. 
This results is the quotient space 
$\mathcal{SO}(d)_{\mathcal{D}} = 
\mathcal{SO}(d)/\mathcal{D}$ of these 
equivalence classes with $\bU = \bA$ in 
$\mathcal{SO}(d)_{\mathcal{D}} \iff 
\mathcal{D}(\bU,\bA) = 0$.\\

\begin{lemma} \label{lemma2}
The function $\calJ_n(\btheta) = 
\sum_{k=1}^{d-1} \dCor_n(\psi(\bSnk(\btheta)),
\psi(\bSnkplus(\btheta)))$ is Lipschitz 
continuous in $\btheta: \bU(\btheta) \in 
\mathcal{SO}(d)_{\mathcal{D}}$.
\end{lemma}
\begin{proof}
First note that the bowl transform is Lipschitz continuous as it is continuously differentiable with bounded derivatives. Additionally, compositions of Lipschitz continuous functions remain Lipschitz continuous. Therefore, the only relevant difference with the proof of the similar Lemma A.2 in \cite{dCovICAmatteson2017} is our use of $\dCor_n$ instead of $\dCov_n$. So for our Lemma to hold, we need to show that $\dCor_n$ is also Lipschitz continuous in our setting.

For this we check whether it has a bounded first derivative. We can use the quotient rule for this:
{\footnotesize
\begin{align} \label{eq:dCornlipschitz}
  \frac{d}{d\btheta}(\dCor_n&(\psi(\bSnk(\btheta)),\psi(\bSnkplus(\btheta)))) = \frac{d}{d\btheta}\left(\frac{\dCov_n(\psi(\bSnk(\btheta),\psi(\bSnkplus(\btheta))))}{\sqrt{\dVar_n(\psi(\bSnk(\btheta)))\dVar_n(\psi(\bSnkplus(\btheta)))}}\right)\\=  &\frac{\frac{d}{d\btheta}(\dCov_n(\psi(\bSnk(\btheta)),\psi(\bSnkplus(\btheta))))}{\dStd_n(\psi(\bSnk(\btheta)))\dStd_n(\psi(\bSnkplus(\btheta)))} \nonumber \\
  &- \left(\frac{\frac{d}{d\btheta}(\dVar_n(\psi(\bSnk(\btheta))))}{2 \,\dVar_n(\psi(\bSnk(\btheta)))} + \frac{\frac{d}{d\btheta}(\dVar_n(\psi(\bSnkplus(\btheta))))}{2 \,\dVar_n(\psi(\bSnkplus(\btheta)))}\right)\; \dCor_n(\psi(\bSnk(\btheta)),\psi(\bSnkplus(\btheta)))\,. \nonumber
\end{align}}
\noindent
\begin{itemize}
\item The numerators are all bounded as $\dCov_n$ and $\dVar_n$ are Lipschitz continuous in $\btheta$ \citep{dCovICAmatteson2017}, $\psi$ is Lipschitz, and $\dCor_n$ is bounded by 1.
\item For the denominators we need a lower bound to upper bound the whole expression in \eqref{eq:dCornlipschitz}. For this we use the extreme value theorem that states that if a function is continuous on a closed and bounded interval that it attains its minimum. We apply this to
\begin{align*}
   m_1 &= \min_{\ell \subseteq \{ 1,2,\dots,d\}} \min_{\btheta \in [0,2\pi]} \dStd_n(\psi(\bSnl(\btheta)))\\
   m_2 &= \min_{\ell \subseteq \{ 1,2,\dots,d\}} \min_{\btheta \in [0,2\pi]} \dVar_n(\psi(\bSnl(\btheta)))
\end{align*}
so $m_1$ and $m_2$ exist. Additionally we have that $m_1,m_2\geq0$ are zero if and only if every sample observation in $\psi(\bSnl(\btheta))$ is identical (Theorem 12.2 in \cite{szekely2023energy}), an event that happens with probability zero in the ICA setting (note that $\psi$ is a bijection). Hence we obtain fixed $m_1, m_2 > 0$ and we can use these to lower bound the denominators of \eqref{eq:dCornlipschitz}. Therefore Equation \eqref{eq:dCornlipschitz} is upper bounded, and hence $\dCor_n$ is Lipschitz continuous in $\btheta$ as it has a bounded first derivative.
\end{itemize}
This finishes the proof of Lemma \ref{lemma2}.
\end{proof}

\begin{lemma} \label{lemma3}
For $\calJ(\btheta) = \sum_{k=1}^{d-1} \dCor(\psi(\bS_{ k}(\btheta)),\psi(\bS_{ k_{+}}(\btheta)))$ it holds that
\begin{equation*}
   \sup_{\btheta:\bU(\btheta) \in  \mathcal{SO}(d)_{\mathcal{D}}} |\calJ_n(\btheta) -\calJ(\btheta)| \overset{a.s.}{\to} 0 \quad \text{ as } \quad n \to \infty\,.
\end{equation*}
\end{lemma}
    
\begin{proof}
We proceed as in \cite{dCovICAmatteson2017}. In their proof they use the Arzelá-Ascoli theorem and the result of Lemma \ref{lemma1} to conclude that it is sufficient to show that
\begin{equation*}
   \lim_{c \to \infty} \overline{\lim_n} \ m_{\frac{1}{c}}(\calJ_n) \overset{a.s.}{=} 0,
\end{equation*}
where
\begin{equation*}
   m_{\frac{1}{c}}(\calJ_n) = \sup \{ |\calJ_n(\btheta) - \calJ_n(\bphi)|: \bU(\btheta),\bU(\bphi) \in \mathcal{SO}(d)_{\mathcal{D}}, ||\bU(\btheta) - \bU(\bphi)||_F < 1/c\}.
\end{equation*}  
Their proof centers around bounding $|\calJ_n(\btheta) -\calJ_n(\bphi)|$ for $\btheta, \bphi \in \Theta$. They find the following (without transforming the random variables, and using $\dCov_n$ instead of $\dCor_n$):
\begin{align*}
   &\hspace{1cm} |\calJ_n(\btheta) -\calJ_n(\bphi)| \leqslant  4 \sum_{\ell = 1}^{d-1} \left( \frac{2}{n} \sum_{i=1}^n || \bs_i(\btheta) -\bs_i(\bphi)|| \right) \times\\ 
   &\left( \binom{n}{2}^{-1}  \sum_{i<j} ||\bs_i(\btheta) - \bs_j(\btheta)|| +  \binom{n}{2}^{-1} \sum_{i<j} ||\bs_i(\bphi) - \bs_j(\bphi)||\right)
\end{align*}
Using the bowl transform does not alter the form of this upper bound, but using $\dCor_n$ does. Every observation is scaled, and we end up with an upper bound of the following form:
\begin{align} \label{eq:dCorLemma3Bound}
  4 \sum_{\ell = 1}^{d-1}& \left( \underbrace{\frac{2}{n} \sum_{i=1}^n || \frac{\psi(\bs_i(\btheta))}{\sigma_n^{(\ell)}(\btheta)} -\frac{\psi(\bs_i(\bphi))}{\sigma_n^{(\ell)}(\bphi)}||}_{(1)} \right) \times \\ 
  &\left( \underbrace{\frac{1}{\sigma_n^{(\ell)}(\btheta)} \binom{n}{2}^{-1} \sum_{i<j} ||\psi(\bs_i(\btheta)) - \psi(\bs_j(\btheta))|| + \frac{1}{\sigma_n^{(\ell)}(\bphi)} \binom{n}{2}^{-1} \sum_{i<j} ||\psi(\bs_i(\bphi)) - \psi(\bs_j(\bphi))||}_{(2)}\right) \nonumber
\end{align}
where
\begin{align*}
   \sigma_n^{(\ell)}(\btheta) &= \sqrt{\dVar_n(\psi(\bSnl(\btheta))) \;\dVar_n(\psi(\bSnlplus(\btheta)))},\\
   \sigma_n^{(\ell)}(\bphi) &= \sqrt{\dVar_n(\psi(\bSnl(\bphi))) \;\dVar_n(\psi(\bSnlplus(\bphi)))}.
\end{align*}
To bound the second factor (2) of Equation \eqref{eq:dCorLemma3Bound}, we use the following:
\begin{itemize}
\item $\left[ \binom{n}{2}^{-1} \sum_{i<j} ||\psi(\bs_i(\btheta)) - \psi(\bs_j(\btheta))|| \right] \overset{a.s.}{\to} \E||\psi(\bs(\btheta)) -\psi(\bs^{\prime}(\btheta))||$ and\\ $\left[ \binom{n}{2}^{-1} \sum_{i<j} ||\psi(\bs_i(\bphi)) - \psi(\bs_j(\bphi))|| \right] \overset{a.s.}{\to} \E||\psi(\bs(\bphi)) -\psi(\bs^{\prime}(\bphi))||$\\ by the SLLN for U-statistics and the boundedness of $\psi$.
\item If we define $\sigma^{(\ell)}(\btheta) = \sqrt{\dVar(\psi(\bS_{\ell}(\btheta))) \;\dVar(\psi(\bS_{\ell +}(\btheta)))}$, then we have $\sigma_n^{(\ell)}(\btheta) \overset{a.s.}{\to} \sigma^{(\ell)}(\btheta)$ and $\sigma_n^{(\ell)}(\bphi) \overset{a.s.}{\to} \sigma^{(\ell)}(\bphi)$ by almost sure convergence of the distance covariance \citep{szekely2007dcor} and the continuous mapping theorem. Here $\sigma^{(\ell)}(\btheta)$ and $\sigma^{(\ell)}(\bphi)$ are strictly greater than zero if the random variables are not degenerate, which is satisfied in the ICA setting.
\item Hence (2) converges to \[ \frac{\E||\psi(\bs(\btheta)) -\psi(\bs^{\prime}(\btheta))||}{\sigma^{(\ell)}(\btheta)} + \frac{\E||\psi(\bs(\bphi)) -\psi(\bs^{\prime}(\bphi))||}{\sigma^{(\ell)}(\bphi)} > 0\]
which behaves well and is bounded.
\end{itemize}
To bound the first factor (1), we proceed as follows:
\begin{align*}
   \frac{2}{n} \sum_{i=1}^n & || \frac{\psi(\bs_i(\btheta))}{\sigma_n^{(\ell)}(\btheta)} -\frac{\psi(\bs_i(\bphi))}{\sigma_n^{(\ell)}(\bphi)}|| = \frac{2}{n} \sum_{i=1}^n || \frac{\psi(\bs_i(\btheta))}{\sigma_n^{(\ell)}(\btheta)} - \frac{\psi(\bs_i(\bphi))}{\sigma_n^{(\ell)}(\btheta)}+ \frac{\psi(\bs_i(\bphi))}{\sigma_n^{(\ell)}(\btheta)}
   -\frac{\psi(\bs_i(\bphi))}{\sigma_n^{(\ell)}(\bphi)}||\\
   &\leqslant \frac{1}{\sigma_n^{(\ell)}(\btheta)} \underbrace{\frac{2}{n} \sum_{i=1}^n || \psi(\bs_i(\btheta)) - \psi(\bs_i(\bphi)) ||}_{(1.a)} + \underbrace{\frac{2\frac{1}{n} \sum_{i=1}^n ||\psi(\bs_i(\bphi))||}{\sigma_n^{(\ell)}(\btheta) \sigma_n^{(\ell)}(\bphi)} \left| \sigma_n^{(\ell)}(\bphi) - \sigma_n^{(\ell)}(\btheta)\right|}_{(1.b)}
\end{align*}
\begin{itemize}
\item The factor $(1.a)$ has the same form as the factors in \cite{dCovICAmatteson2017}, where they show that it shrinks to zero for $||\bU(\btheta) - \bU(\bphi)||$ going to zero.
\item For $(1.b)$, $\frac{\frac{1}{n} \sum_{i=1}^n ||\psi(\bs_i(\bphi))||}{\sigma_n^{(\ell)}(\btheta) \sigma_n^{(\ell)}(\bphi)}$ behaves well as before, and converges to a constant. We however have to check that $\left| \sigma_n^{(\ell)}(\bphi) - \sigma_n^{(\ell)}(\btheta)\right|$ goes to zero for $||\bU(\btheta) - \bU(\bphi)||$ going to zero:
{\footnotesize
\begin{align*}
   \left|\sigma_n^{(\ell)}(\btheta) - \sigma_n^{(\ell)}(\bphi)\right| &= \Big|\sqrt{\dVar_n(\psi(\bSnl(\btheta))) \;\dVar_n(\psi(\bSnlplus(\btheta)))}\\
   &\qquad - \sqrt{\dVar_n(\psi(\bSnl(\bphi))) \;\dVar_n(\psi(\bSnlplus(\bphi)))}\Big|\\
   &\leqslant \sqrt{\dVar_n(\psi(\bSnl(\btheta)))} \;\Big|\sqrt{\dVar_n(\psi(\bSnlplus(\btheta)))} - \sqrt{\dVar_n(\psi(\bSnlplus(\bphi)))}\Big| \\
   &\qquad + \sqrt{\dVar_n(\psi(\bSnlplus(\bphi)))}\Big|\sqrt{\dVar_n(\psi(\bSnl(\btheta)))} - \sqrt{\dVar_n(\psi(\bSnl(\bphi)))}\Big|
\end{align*}}
Here the factors $\sqrt{\dVar_n(\psi(\bSnl(\btheta)))}$ and $\sqrt{\dVar_n(\psi(\bSnlplus(\bphi)))}$ behave well and converge to a positive constant. For the factor $\Big|\sqrt{\dVar_n(\psi(\bSnl(\btheta)))} - \sqrt{\dVar_n(\psi(\bSnl(\bphi)))}\Big|$ and the other similar factor we have
\begin{align*}
   \big|\sqrt{\dVar_n(\psi(\bSnl(\btheta)))} & - 
   \sqrt{\dVar_n(\psi(\bSnl(\bphi)))}\big| \\ 
   & = \frac{\big|\dVar_n(\psi(\bSnl(\btheta))) - 
   \dVar_n(\psi(\bSnl(\bphi)))\big|}
   {\big|\sqrt{\dVar_n(\psi(\bSnl(\btheta)))} + 
   \sqrt{\dVar_n(\psi(\bSnl(\bphi)))}\big|},
\end{align*}
where the denominator behaves well. For the numerator we can use the result of Lemma A.3 of \cite{dCovICAmatteson2017}, stating that this shrinks to zero when $||\bU(\btheta) - \bU(\bphi)||$ goes to zero.
\end{itemize}
Putting this together, (1) shrinks to zero, and hence also the initial Equation \eqref{eq:dCorLemma3Bound}:
\begin{equation*}
    \sup_{||\bU(\btheta)-\bU(\bphi)||_F<\delta} |\calJ_n(\btheta) -\calJ_n(\bphi)| \leqslant \alpha \delta
\end{equation*}
for some constant $\alpha$.

With this result, we can, as in \cite{dCovICAmatteson2017}, apply the Arzelá-Ascoli theorem as $\mathcal{SO}(d)_{\mathcal{D}}$ is separable, which yields:
\begin{equation*}
    \sup_{\btheta: \bU(\btheta) \in \mathcal{SO}(d)_{\mathcal{D}}} |\calJ_n(\btheta) - \calJ(\btheta)| \ \overset{a.s.}{\to} \ 0 \ \text{ for } \ n \to \infty.
\end{equation*}
\end{proof}

\textbf{Proof of Proposition \ref{prop:asconv}.}
Using these three extended lemmas, we can now prove 
Proposition~\ref{prop:asconv}. First note that 
$\forall n$ we have 
$\calJ_n(\btheta_0) \geq \calJ_n(\bhtheta_n)$ and $\calJ(\bhtheta_n) \geq \calJ(\btheta_0)$ because $\bhtheta_n$ is the minimum of $\calJ_n$ and $\btheta_0$ is the minimum of $\calJ$. This yields
\begin{equation*}
    \calJ_n(\btheta_0) - \calJ(\btheta_0) \geq \calJ_n(\bhtheta_n) - \calJ(\btheta_0) \geq \calJ_n(\bhtheta_n) - \calJ(\bhtheta_n).
\end{equation*}
Therefore:
\begin{align*}
    | \calJ_n(\bhtheta_n) - \calJ(\btheta_0) | &\leqslant \max \{|\calJ_n(\btheta_0) - \calJ(\btheta_0)|,|\calJ_n(\bhtheta_n) - \calJ(\bhtheta_n)| \} \\
    &\leqslant \sup_{\bphi: \bU(\bphi) \in \mathcal{SO}(d)_{\mathcal{D}}} |\calJ_n(\bphi) - \calJ(\bphi)|.
\end{align*}
Now Lemma \ref{lemma3} states that this supremum goes to zero almost surely, hence $\calJ_n(\bhtheta_n) \overset{a.s.}{\to} \calJ(\btheta_0)$ for $n \to \infty$. Additionally, we have that $\mathcal{SO}(d)_{\mathcal{D}}$ is compact, therefore the minima of $\calJ$ and $\calJ_n$ exist in  $\mathcal{SO}(d)_{\mathcal{D}}$. Also, the argmin mapping is continuous, yielding $\bU(\bhtheta_n) \overset{a.s.}{\to} \bU(\btheta_0)$ for $n \to \infty$ and $\bU(\btheta_0) \in \mathcal{SO}(d)_{\mathcal{D}}$. If now $\btheta_0$ is in $\overline{\Theta}$, a sufficiently large compact subset of $\Theta$, the continuous mapping theorem gives us $\bhtheta_n \overset{a.s.}{\to} \btheta_0$.
\qed

\subsection{Proof of Proposition \ref{prop:rootncons}}

We follow the lines of Theorem 2.2 of \cite{dCovICAmatteson2017}.

We use the following subset $\Omega$ of ${\cal SO}(d)$ introduced by \cite{chen2005consistent}. First, each row of $\bU$ has Euclidean norm 1. Second, the element with maximal modulus in each row of $\bU$ is positive. Finally, the rows of $\bU$ are sorted according to the partial order $\prec$ that is given by: $\forall \; a,b \in \mathbb{R}^d$, $a \prec b$ if and only if there exists $ k \in \{1,\ldots,d\}$ such that $a_k < b_k$ and $a_j = b_j$ for $j \in \{1,\ldots,k-1\}$. This construction gets rid of identifiability issues as each $\bU$ along with all of its signed permutations corresponds to a single element in $\Omega$.

Now we introduce a path $\gamma$ between two points on the unit ball in $\mathbb{R}^d$ following \cite{dCovICAmatteson2017}. Consider a unit ball in $\mathbb{R}^d$ centered at $A$, which contains the points $B$ and $C$, and let the angle $\xi > 0$ denote the smallest value such that $\cos(\xi) = \langle \vec{AB}, \vec{AC} \rangle$. For $\tau \in \mathbb{R}$, let 
\[
\gamma(\tau) = \cos(\tau)\vec{AB} + \sin(\tau)\vec{AD}
\]
denote a path from $\vec{AB}$ to $\vec{AC}$, in which $\vec{AD}$ is a unit tangent vector at $B$ such that $\vec{AB}, \vec{AC},$ and $\vec{AD}$ are on the same hyperplane. Then, $\gamma(0) = B,\ \gamma(\xi) = C$, and 
\[
\left\| \frac{\partial}{\partial\tau} \gamma(\tau) \right\| = 1; \quad \text{further, } \|\gamma(\tau_2) - \gamma(\tau_1)\| \leqslant |\tau_2 - \tau_1|.
\]

By definition, each row of $\bU \in \Omega$ is on the unit ball in $\mathbb{R}^d$. Let $\xi_1, \ldots, \xi_d$ denote the angles between the corresponding rows of $\bU_0 = \bU_{\theta_0}$ and $\widehat{\bU} =\bU_{\tilde{\theta}_n}$.

Let $\hat{\eta} = \sqrt{ \sum_{k=1}^d \xi_k^2 }$\,, and note that $\| \widehat{\bU} - \bP_{\pm} \bU_0 \| = o_p(1)$ implies $\hat{\eta} = o_p(1)$.

Now let $\gamma : \mathbb{R} \rightarrow \mathbb{R}^{d\times d}$ be such that $\gamma(0) = \bU_0$ and $\gamma(\hat{\eta}) = \widehat{\bU}$, by considering $\gamma(\cdot)$ for the $k$th rows, as described above but rescaled by $\xi_k / \hat{\eta}$. Then we similarly note that 
\[
\left\| \frac{\partial}{\partial\tau} \gamma(\tau) \right\| = \sqrt{ \sum_{k=1}^d \left( \frac{\xi_k}{\hat{\eta}} \right)^2 } = 1, \quad \text{and } \| \gamma(\tau_2) - \gamma(\tau_1) \| \leqslant |\tau_2 - \tau_1|.
\]
As such, for $\| \widehat{\bU} - P_{\pm} \bU_0 \| \leqslant \hat{\eta}$ and sufficiently small $\tau \geq 0$, we note $\gamma(\tau) \in \Omega$.

We will now use
$\mathcal{J}_n(\gamma(\tau))$ to characterize the objective function, where $\gamma(0) = \bU_0$ and $\gamma(\hat{\eta}) = \bUh$. We consider
$$\mathcal{J}_n(\gamma(\tau)) = \sum_{k=1}^{d-1} \dCor_n(\psi(\bSnk(\gamma(\tau))),\psi(\bSnkplus(\gamma(\tau))))\;.$$
Now we consider the first derivative of this objective with respect to $\tau$, $\frac{\partial}{\partial \tau} \mathcal{J}_n(\gamma(\tau))$ and use Taylor's theorem with the mean-value theorem to conclude that there exists a $\bar{\tau} \in [0, \hat{\eta}]$ for which
\[
0 = \frac{\partial}{\partial \tau} \mathcal{J}_n(\gamma(\hat{\eta})) = \frac{\partial}{\partial \tau} \mathcal{J}_n(\gamma(0)) + \hat{\eta} \frac{\partial^2}{\partial \tau^2} \mathcal{J}_n(\gamma(\bar{\tau})),
\]
which implies
\[
\hat{\eta} = - \frac{ \frac{\partial}{\partial \tau} \mathcal{J}_n(\gamma(0)) }{ \frac{\partial^2}{\partial \tau^2} \mathcal{J}_n(\gamma(\bar{\tau})) }. 
\]

First consider the numerator. We know $\mathcal{J}_n$ consists of $d-1$ distance correlations, each of which can be written as a distance covariance (i.e., a sum of three U-statistics) multiplied by a scale factor.  The scale factor converges in probability to its population counterpart as a consequence of the continuous mapping theorem and the a.s. convergence of the distance variance obtained in \cite{szekely2007dcor}, and is thus $\Op(1)$. The derivative of the scale term is also $\Op(1)$. This follows from the chain rule, as the resulting expression is a continuous function of the scale itself, the distance variance, and the first derivative of the distance variance, all of which converge in probability to their population counterparts.\\

We can thus conclude that the scale factors do not affect the  convergence rate. This leaves the derivative of the distance covariance and the distance covariance itself as deciding factors. The first is $\Op(1/\sqrt{n})$ by \cite{dCovICAmatteson2017}, whereas the second is $\Op(1/\sqrt{n})$ by \cite{szekely2007dcor} under the model ICA assumptions. We conclude that each of the distance correlation terms in the numerator is $\Op(1/\sqrt{n})$, and thus the numerator itself as well. Note that we can obtain a similar result under model misspecification provided that $\E\left[\frac{\partial}{\partial \tau} \mathcal{J}_n(\gamma(\tau))\big|_{\tau = 0}\right] = \op(1/\sqrt{n})$.\\

Now consider the denominator $ \frac{\partial^2}{\partial \tau^2} \mathcal{J}_n(\gamma(\bar{\tau}))$. We need to show that this quantity is bounded from below asymptotically, so that it does not interfere with the $\Op(1/\sqrt{n})$ convergence of the numerator. As $\mathcal{J}_n(\gamma(\bar{\tau}))$ is a sum of distance correlations, $\frac{\partial^2}{\partial \tau^2} \widetilde{\mathcal{J}}_n(\gamma(\tau))$  depends on the distance covariances and their first two derivatives, as well as the scale factors and their first two derivatives. Note that all the involved quantities are stochastically bounded. Of the six terms coming out of the second derivative of each distance correlation, all but one contain either 
$\dCov_n$ or $\dCov_n'$, both of which converge to zero in probability under the ICA model. The final term is $\frac{\partial^2}{\partial \tau^2} \dCov_n(\gamma(\tau))$ multiplied by a scale factor. The scale factor converges in probability to its population counterpart (which is a strictly positive number) by the continuous mapping theorem and \cite{szekely2007dcor}. The problem is therefore reduced to bounding the sum of the second derivatives of all $\dCov$-terms from below, which was done in \cite{dCovICAmatteson2017}.

As a result, we obtain

\[
\frac{\partial^2}{\partial \tau^2} \mathcal{J}_n(\gamma(\bar{\tau})) = \min_{\frac{\partial}{\partial \tau} \gamma(0), \frac{\partial^2}{\partial \tau^2} \gamma(0)} \frac{\partial^2}{\partial \tau^2} \mathcal{J}(\gamma(\tau)) + o_p(1),
\]

where $\min_{\frac{\partial}{\partial \tau} \gamma(0), \frac{\partial^2}{\partial \tau^2} \gamma(0)} \frac{\partial^2}{\partial \tau^2} \mathcal{J}(\gamma(\tau)) >0$ given that $\gamma(0)$ is the unique global minimizer and by differentiability and compactness. Putting the results for the numerator and the denominator together, we obtain 
\[\hat{\eta} = \Op(n^{-1/2}).\]
This ends the proof.

\section{Tables with detailed results of 
Section \ref{sec:simulationresults}}
\subsection{Uncontaminated data}

\begin{table}[H]
\vspace{10mm}
\centering
\scalebox{0.9}{\begin{tabular}{rrrrrrrrr}
  \hline
 & RICA & RICA no sweeps & dCovICA 
 & FastICA & Infomax & JADE & RADICAL \\ 
  \hline
a & 11.57 & 11.52 & 2.66 & 5.24 & 2.57 & 3.71 & 4.15 \\ 
  b & 6.04 & 6.03 & 2.67 & 3.89 & 2.94 & 4.83 & 4.10 \\ 
  c & 5.01 & 5.05 & 1.78 & 2.36 & 2.06 & 1.74 & 2.93 \\ 
  d & 16.08 & 16.01 & 4.93 & 5.35 & 4.17 & 5.46 & 7.44 \\ 
  e & 3.80 & 3.80 & 1.55 & 6.53 & 3.70 & 4.17 & 2.83 \\ 
  f & 1.69 & 1.69 & 1.41 & 2.02 & 1.73 & 2.55 & 4.08 \\ 
  g & 2.29 & 2.29 & 1.53 & 1.73 & 1.58 & 1.64 & 5.01 \\ 
  h & 10.37 & 10.33 & 4.14 & 4.14 & 4.01 & 4.44 & 8.59 \\ 
  i & 17.31 & 17.42 & 8.24 & 7.51 & 6.69 & 7.17 & 20.93 \\ 
  j & 3.77 & 3.77 & 1.63 & 50.55 & 46.81 & 6.94 & 3.59 \\ 
  k & 4.58 & 4.58 & 2.32 & 41.36 & 36.52 & 14.84 & 4.47 \\ 
  l & 5.28 & 5.28 & 4.20 & 38.01 & 35.25 & 26.53 & 8.44 \\ 
  m & 3.13 & 4.24 & 28.90 & 6.44 & 4.81 & 3.32 & 1.98 \\ 
  n & 8.22 & 9.27 & 14.22 & 41.28 & 45.85 & 12.10 & 4.93 \\ 
  o & 9.54 & 9.53 & 8.59 & 8.36 & 7.08 & 4.80 & 13.18 \\ 
  p & 2.45 & 2.45 & 1.84 & 20.44 & 14.37 & 3.73 & 2.76 \\ 
  q & 8.06 & 8.66 & 7.84 & 16.83 & 23.49 & 42.86 & 6.14 \\ 
  r & 12.48 & 12.58 & 8.32 & 43.23 & 44.99 & 24.99 & 20.58 \\ 
  \hline
  mean & 7.32 & 7.47 & 5.93 & 16.96 & 16.03 & 9.77 & 7.01 \\ 
   \hline
\end{tabular}}
\caption{Amari error ($\times 100$) for 
  $d=2$, no contamination.}
\end{table}

\begin{table}[H] 
\vspace{10mm}
\centering
\scalebox{0.9}{\begin{tabular}{rrrrrrrrr}
  \hline
 & RICA & RICA no sweeps & dCovICA 
 & FastICA & Infomax & JADE & RADICAL \\ 
  \hline
a & 10.27 & 10.53 & 2.77 & 2.36 & 2.48 & 3.97 & 3.87 \\ 
  b & 5.85 & 5.85 & 4.74 & 2.95 & 3.03 & 4.77 & 4.19 \\ 
  c & 3.87 & 4.42 & 8.19 & 2.09 & 2.08 & 1.83 & 2.15 \\ 
  d & 14.72 & 15.53 & 6.58 & 4.50 & 4.21 & 5.57 & 7.58 \\ 
  e & 2.71 & 2.70 & 1.70 & 4.03 & 3.52 & 4.15 & 2.35 \\ 
  f & 2.33 & 10.16 & 16.49 & 1.73 & 1.78 & 2.57 & 1.86 \\ 
  g & 2.59 & 9.28 & 12.20 & 1.54 & 1.54 & 1.61 & 2.09 \\ 
  h & 8.57 & 9.32 & 12.13 & 4.10 & 4.12 & 4.62 & 13.64 \\ 
  i & 17.41 & 19.74 & 18.80 & 6.34 & 6.84 & 7.15 & 35.48 \\ 
  j & 2.68 & 3.75 & 3.38 & 28.04 & 26.46 & 9.20 & 2.15 \\ 
  k & 3.21 & 5.83 & 2.73 & 30.10 & 29.58 & 17.87 & 3.64 \\ 
  l & 5.17 & 8.05 & 4.48 & 37.13 & 36.89 & 34.21 & 9.57 \\ 
  m & 11.29 & 16.20 & 26.61 & 6.05 & 4.85 & 3.47 & 6.64 \\ 
  n & 28.95 & 35.99 & 41.21 & 40.94 & 42.12 & 18.40 & 30.64 \\ 
  o & 15.41 & 19.17 & 23.32 & 7.52 & 6.81 & 4.78 & 35.22 \\ 
  p & 3.26 & 12.54 & 10.05 & 18.64 & 14.64 & 3.93 & 2.19 \\ 
  q & 19.24 & 26.57 & 33.87 & 22.18 & 33.24 & 45.75 & 27.90 \\ 
  r & 14.22 & 19.02 & 16.52 & 41.28 & 41.66 & 32.70 & 35.32 \\ 
  \hline
  mean & 9.54 & 13.04 & 13.65 & 14.53 & 14.77 & 11.47 & 12.58 \\ 
   \hline
\end{tabular}}
\caption{Amari error ($\times 100$) for $d=4$, 
  no contamination.}
\end{table}

\begin{table}[H] 
\vspace{10mm}
\centering
\scalebox{0.9}{\begin{tabular}{rrrrrrrrr}
  \hline
 & RICA & RICA no sweeps & dCovICA 
 & FastICA & Infomax & JADE & RADICAL \\ 
  \hline
a & 6.69 & 7.14 & 1.91 & 1.68 & 1.65 & 2.81 & 2.67 \\ 
  b & 3.86 & 3.77 & 5.64 & 2.00 & 2.08 & 3.36 & 2.95 \\ 
  c & 1.70 & 7.52 & 13.67 & 1.49 & 1.41 & 1.26 & 1.34 \\ 
  d & 9.71 & 11.47 & 5.42 & 2.86 & 2.75 & 4.14 & 5.10 \\ 
  e & 2.00 & 2.01 & 1.18 & 2.86 & 2.41 & 2.98 & 1.63 \\ 
  f & 1.61 & 16.33 & 17.25 & 1.18 & 1.21 & 1.72 & 1.33 \\ 
  g & 1.83 & 16.36 & 15.94 & 1.07 & 1.06 & 1.11 & 1.40 \\ 
  h & 4.75 & 12.15 & 19.15 & 2.73 & 2.75 & 3.19 & 13.36 \\ 
  i & 14.80 & 21.88 & 26.49 & 4.61 & 4.45 & 5.08 & 37.00 \\ 
  j & 1.91 & 7.14 & 6.50 & 18.43 & 16.96 & 5.09 & 1.86 \\ 
  k & 2.23 & 10.22 & 5.18 & 24.38 & 22.06 & 14.03 & 2.75 \\ 
  l & 3.67 & 10.16 & 5.29 & 32.63 & 32.28 & 29.77 & 6.29 \\ 
  m & 13.28 & 18.03 & 25.63 & 3.84 & 3.36 & 2.45 & 15.98 \\ 
  n & 36.86 & 38.46 & 39.15 & 37.49 & 36.89 & 11.61 & 36.46 \\ 
  o & 12.48 & 20.11 & 22.65 & 5.33 & 4.96 & 3.49 & 37.46 \\ 
  p & 1.99 & 18.06 & 16.46 & 12.31 & 8.95 & 2.69 & 1.50 \\ 
  q & 30.41 & 34.45 & 35.57 & 17.54 & 30.49 & 41.89 & 33.53 \\ 
  r & 12.37 & 24.65 & 20.19 & 39.62 & 39.37 & 29.65 & 37.69 \\ 
  \hline
  mean & 9.01 & 15.55 & 15.74 & 11.78 & 11.95 & 9.24 & 13.35 \\ 
   \hline
\end{tabular}}
\caption{Amari error ($\times 100$) for $d=6$, 
   no contamination.}
\end{table}

\subsection{Clustered contamination}

\begin{table}[H] 
\vspace{10mm}
\centering
\scalebox{0.9}{\begin{tabular}{rrrrrrrrr}
   \hline
 & RICA & RICA no sweeps & dCovICA 
 & FastICA & Infomax & JADE & RADICAL \\ 
  \hline
a & 7.77 & 7.96 & 54.98 & 80.17 & 77.88 & 64.78 & 85.76 \\ 
  b & 4.51 & 4.45 & 55.78 & 84.11 & 83.43 & 84.71 & 84.87 \\ 
  c & 2.71 & 2.71 & 56.72 & 85.45 & 84.43 & 86.64 & 87.35 \\ 
  d & 15.65 & 14.36 & 56.04 & 83.35 & 82.20 & 81.16 & 85.62 \\ 
  e & 3.51 & 3.51 & 55.00 & 85.41 & 85.13 & 84.44 & 3.32 \\ 
  f & 1.83 & 1.83 & 57.00 & 85.36 & 84.45 & 86.41 & 3.99 \\ 
  g & 1.89 & 1.89 & 57.04 & 84.91 & 84.07 & 86.06 & 9.52 \\ 
  h & 22.27 & 21.85 & 56.60 & 85.09 & 84.11 & 86.36 & 87.04 \\ 
  i & 38.91 & 37.82 & 56.54 & 85.34 & 84.37 & 86.35 & 86.98 \\ 
  j & 5.96 & 5.96 & 56.30 & 86.24 & 85.56 & 86.85 & 7.15 \\ 
  k & 3.91 & 3.91 & 56.18 & 85.14 & 84.21 & 85.81 & 77.49 \\ 
  l & 5.22 & 5.32 & 56.28 & 85.28 & 84.42 & 86.27 & 86.08 \\ 
  m & 8.28 & 10.54 & 56.58 & 85.39 & 84.41 & 86.50 & 87.30 \\ 
  n & 22.63 & 21.57 & 56.39 & 85.45 & 84.46 & 86.83 & 87.13 \\ 
  o & 25.79 & 24.16 & 56.46 & 85.16 & 84.15 & 86.44 & 87.04 \\ 
  p & 2.91 & 3.43 & 56.48 & 85.29 & 84.28 & 86.40 & 84.84 \\ 
  q & 16.82 & 16.31 & 56.25 & 84.88 & 84.04 & 86.12 & 86.32 \\ 
  r & 14.14 & 13.65 & 56.28 & 85.22 & 84.26 & 86.41 & 86.97 \\ 
  \hline
  mean & 11.37 & 11.18 & 56.27 & 84.85 & 83.88 & 84.70 & 68.04 \\ 
   \hline
\end{tabular}}
\caption{Amari error ($\times 100$) for $d=2$, 
  clustered contamination.}
\end{table}

\begin{table}[H] 
\vspace{10mm}
\centering
\scalebox{0.9}{\begin{tabular}{rrrrrrrrr}
  \hline
 & RICA & RICA no sweeps & dCovICA 
 & FastICA & Infomax & JADE & RADICAL \\ 
  \hline
a & 10.74 & 16.16 & 35.13 & 41.39 & 38.48 & 33.83 & 44.67 \\ 
  b & 4.58 & 12.20 & 36.11 & 42.33 & 41.35 & 43.36 & 45.11 \\ 
  c & 16.76 & 25.71 & 36.47 & 42.98 & 41.94 & 45.24 & 47.63 \\ 
  d & 18.36 & 19.73 & 37.22 & 43.29 & 41.98 & 41.90 & 46.06 \\ 
  e & 2.60 & 2.77 & 31.55 & 41.88 & 41.10 & 43.52 & 43.60 \\ 
  f & 5.73 & 26.60 & 37.17 & 42.65 & 41.23 & 45.44 & 51.37 \\ 
  g & 2.40 & 19.04 & 35.53 & 41.88 & 40.70 & 44.90 & 51.20 \\ 
  h & 49.50 & 49.57 & 38.95 & 43.42 & 42.40 & 45.64 & 52.02 \\ 
  i & 50.92 & 49.21 & 40.78 & 45.11 & 44.89 & 47.38 & 52.98 \\ 
  j & 2.95 & 5.72 & 32.74 & 47.89 & 46.19 & 45.37 & 33.39 \\ 
  k & 3.17 & 7.28 & 34.18 & 48.71 & 46.95 & 46.45 & 44.63 \\ 
  l & 6.41 & 12.28 & 35.82 & 49.60 & 48.92 & 49.34 & 47.93 \\ 
  m & 47.02 & 46.79 & 40.37 & 43.77 & 42.60 & 44.96 & 51.78 \\ 
  n & 32.46 & 31.90 & 44.84 & 50.30 & 49.44 & 50.54 & 54.41 \\ 
  o & 47.39 & 47.17 & 41.67 & 45.86 & 44.43 & 46.62 & 55.84 \\ 
  p & 5.60 & 17.57 & 36.36 & 47.57 & 45.89 & 46.32 & 49.42 \\ 
  q & 23.72 & 28.59 & 43.22 & 50.07 & 48.81 & 52.24 & 54.90 \\ 
  r & 26.28 & 29.54 & 38.95 & 52.11 & 51.35 & 51.09 & 53.04 \\ 
  \hline
  mean & 19.81 & 24.88 & 37.61 & 45.60 & 44.37 & 45.79 & 48.89 \\ 
   \hline
\end{tabular}}
\caption{Amari error ($\times 100$) for $d=4$, 
   clustered contamination.}
\end{table}

\begin{table}[H] 
\vspace{10mm}
\centering
\scalebox{0.9}{\begin{tabular}{rrrrrrrrr}
  \hline
 & RICA & RICA no sweeps & dCovICA 
 & FastICA & Infomax & JADE & RADICAL \\ 
  \hline
a & 11.04 & 17.12 & 24.35 & 26.69 & 25.17 & 21.39 & 29.33 \\ 
  b & 3.15 & 14.18 & 25.82 & 27.67 & 26.75 & 28.25 & 29.60 \\ 
  c & 37.03 & 40.71 & 28.05 & 27.57 & 26.83 & 29.91 & 30.52 \\ 
  d & 15.01 & 19.14 & 26.49 & 27.90 & 27.13 & 27.28 & 30.02 \\ 
  e & 1.89 & 3.05 & 22.23 & 26.32 & 27.45 & 26.89 & 28.93 \\ 
  f & 38.34 & 41.02 & 35.23 & 27.48 & 26.51 & 30.09 & 32.30 \\ 
  g & 30.99 & 38.41 & 32.91 & 26.90 & 25.61 & 29.90 & 31.11 \\ 
  h & 45.29 & 47.98 & 32.85 & 28.65 & 27.85 & 30.15 & 40.48 \\ 
  i & 45.40 & 46.18 & 36.18 & 30.02 & 29.35 & 31.56 & 46.36 \\ 
  j & 1.98 & 10.18 & 25.85 & 35.90 & 36.47 & 29.51 & 26.87 \\ 
  k & 2.30 & 14.04 & 26.05 & 41.84 & 41.59 & 31.32 & 29.52 \\ 
  l & 6.50 & 21.74 & 26.94 & 42.80 & 42.32 & 37.15 & 35.31 \\ 
  m & 43.36 & 46.10 & 34.10 & 29.20 & 28.39 & 29.85 & 47.41 \\ 
  n & 38.88 & 38.45 & 40.42 & 44.45 & 42.86 & 38.75 & 49.84 \\ 
  o & 43.53 & 45.88 & 36.08 & 30.06 & 29.27 & 30.25 & 47.85 \\ 
  p & 22.04 & 29.43 & 30.21 & 32.22 & 30.55 & 30.13 & 35.94 \\ 
  q & 29.76 & 33.02 & 41.23 & 42.74 & 43.90 & 47.23 & 48.76 \\ 
  r & 31.53 & 32.70 & 33.48 & 44.99 & 44.34 & 42.24 & 47.38 \\ 
  \hline
  mean & 24.89 & 29.96 & 31.03 & 32.97 & 32.35 & 31.77 & 37.09 \\ 
   \hline
\end{tabular}}
\caption{Amari error ($\times 100$) for $d=6$, 
   clustered contamination.}
\end{table}

\subsection{Multiplicative contamination}

\begin{table}[H] 
\vspace{10mm}
\centering
\scalebox{0.9}{\begin{tabular}{rrrrrrrrr}
  \hline
 & RICA & RICA no sweeps & dCovICA 
 & FastICA & Infomax & JADE & RADICAL \\ 
  \hline
a & 10.14 & 10.12 & 65.34 & 64.88 & 74.03 & 75.99 & 30.45 \\ 
  b & 6.07 & 6.11 & 88.35 & 75.43 & 83.39 & 84.03 & 67.03 \\ 
  c & 30.39 & 30.54 & 3.25 & 41.03 & 45.07 & 86.88 & 9.03 \\ 
  d & 18.46 & 18.52 & 78.89 & 72.30 & 80.26 & 81.45 & 69.82 \\ 
  e & 2.20 & 2.20 & 11.32 & 34.18 & 22.31 & 22.13 & 2.60 \\ 
  f & 5.83 & 5.88 & 31.51 & 84.60 & 85.86 & 85.79 & 3.30 \\ 
  g & 6.04 & 6.04 & 3.80 & 45.84 & 57.41 & 86.60 & 8.16 \\ 
  h & 19.53 & 19.52 & 92.78 & 83.37 & 85.78 & 85.85 & 84.30 \\ 
  i & 22.60 & 22.49 & 92.55 & 83.94 & 85.80 & 86.04 & 84.42 \\ 
  j & 4.08 & 4.08 & 2.65 & 84.43 & 85.65 & 85.17 & 2.42 \\ 
  k & 3.73 & 3.73 & 5.82 & 82.60 & 85.26 & 85.44 & 13.00 \\ 
  l & 5.24 & 5.24 & 81.62 & 81.91 & 85.33 & 85.70 & 82.66 \\ 
  m & 49.17 & 49.75 & 94.41 & 85.16 & 86.24 & 86.34 & 5.60 \\ 
  n & 16.96 & 19.24 & 93.27 & 83.41 & 85.78 & 86.11 & 84.36 \\ 
  o & 29.93 & 30.29 & 93.10 & 83.98 & 85.63 & 86.12 & 84.58 \\ 
  p & 3.56 & 3.88 & 77.66 & 85.44 & 86.54 & 86.68 & 2.90 \\ 
  q & 8.78 & 9.88 & 92.74 & 81.52 & 85.24 & 85.85 & 84.52 \\ 
  r & 10.67 & 11.57 & 92.13 & 82.53 & 85.57 & 85.98 & 84.41 \\ 
  \hline
  mean & 14.08 & 14.39 & 61.18 & 74.25 & 77.29 & 81.56 & 44.64 \\ 
   \hline
\end{tabular}}
\caption{Amari error ($\times 100$) for $d=2$, 
   multiplicative contamination.}
\end{table}

\begin{table}[H] 
\vspace{10mm}
\centering
\scalebox{0.9}{\begin{tabular}{rrrrrrrrr}
  \hline
 & RICA & RICA no sweeps & dCovICA 
 & FastICA & Infomax & JADE & RADICAL \\ 
  \hline
a & 7.92 & 9.22 & 38.87 & 41.97 & 47.73 & 54.20 & 7.53 \\ 
  b & 4.96 & 5.36 & 50.90 & 45.70 & 55.36 & 59.78 & 19.52 \\ 
  c & 3.04 & 3.64 & 58.33 & 60.30 & 64.11 & 63.49 & 29.31 \\ 
  d & 12.91 & 15.06 & 46.97 & 44.14 & 54.09 & 58.61 & 30.14 \\ 
  e & 3.05 & 2.98 & 11.53 & 30.48 & 20.40 & 34.25 & 2.48 \\ 
  f & 2.07 & 5.18 & 50.29 & 52.92 & 59.82 & 61.95 & 1.98 \\ 
  g & 2.03 & 3.99 & 50.56 & 59.48 & 62.63 & 62.34 & 2.13 \\ 
  h & 8.70 & 11.55 & 57.37 & 51.47 & 58.63 & 62.77 & 51.82 \\ 
  i & 19.71 & 23.14 & 56.67 & 50.63 & 59.29 & 62.70 & 54.12 \\ 
  j & 2.91 & 4.26 & 20.44 & 52.00 & 58.87 & 59.23 & 2.54 \\ 
  k & 3.41 & 6.49 & 43.83 & 50.55 & 58.22 & 60.06 & 11.33 \\ 
  l & 5.69 & 9.65 & 52.27 & 50.12 & 58.07 & 60.97 & 48.80 \\ 
  m & 10.21 & 13.86 & 57.63 & 53.04 & 59.37 & 62.58 & 41.67 \\ 
  n & 35.50 & 40.34 & 57.19 & 52.17 & 60.05 & 64.87 & 53.52 \\ 
  o & 14.84 & 18.29 & 56.40 & 51.62 & 59.62 & 63.85 & 54.22 \\ 
  p & 3.35 & 11.92 & 56.45 & 52.32 & 59.08 & 62.43 & 22.59 \\ 
  q & 30.71 & 33.76 & 56.16 & 49.95 & 57.66 & 62.36 & 52.23 \\ 
  r & 18.12 & 23.57 & 55.09 & 50.01 & 58.07 & 62.28 & 52.39 \\ 
  \hline
  mean & 10.51 & 13.46 & 48.72 & 49.94 & 56.17 & 59.93 & 29.91 \\ 
   \hline
\end{tabular}}
\caption{Amari error ($\times 100$) for $d=4$, 
  multiplicative contamination.}
\end{table}

\begin{table}[H] 
\vspace{10mm}
\centering
\scalebox{0.9}{\begin{tabular}{rrrrrrrrr}
  \hline
 & RICA & RICA no sweeps & dCovICA 
& FastICA & Infomax & JADE & RADICAL \\ 
  \hline
a & 6.63 & 9.08 & 37.17 & 40.26 & 44.52 & 49.63 & 4.94 \\ 
  b & 3.71 & 6.40 & 46.26 & 41.86 & 51.30 & 56.25 & 12.84 \\ 
  c & 1.82 & 8.01 & 51.31 & 47.13 & 54.28 & 57.63 & 42.22 \\ 
  d & 9.41 & 15.49 & 42.69 & 40.86 & 49.01 & 53.10 & 24.38 \\ 
  e & 2.27 & 2.23 & 15.21 & 29.37 & 20.17 & 36.65 & 1.71 \\ 
  f & 4.36 & 20.62 & 47.26 & 44.13 & 52.59 & 56.90 & 2.10 \\ 
  g & 1.71 & 15.61 & 49.33 & 47.05 & 53.38 & 58.12 & 1.65 \\ 
  h & 27.62 & 31.20 & 49.18 & 43.38 & 52.75 & 57.69 & 47.20 \\ 
  i & 39.89 & 39.45 & 49.93 & 44.24 & 53.56 & 57.44 & 47.08 \\ 
  j & 1.99 & 8.39 & 39.51 & 44.15 & 51.14 & 54.92 & 2.04 \\ 
  k & 2.36 & 12.01 & 46.43 & 44.37 & 50.78 & 55.45 & 12.44 \\ 
  l & 5.07 & 18.26 & 47.30 & 43.50 & 52.47 & 56.64 & 42.88 \\ 
  m & 23.73 & 27.44 & 50.09 & 44.49 & 53.89 & 58.30 & 47.99 \\ 
  n & 39.54 & 40.70 & 50.11 & 43.86 & 52.76 & 57.52 & 46.35 \\ 
  o & 33.12 & 34.27 & 50.29 & 43.89 & 53.38 & 57.57 & 47.31 \\ 
  p & 5.50 & 21.64 & 50.19 & 45.32 & 53.50 & 57.59 & 43.04 \\ 
  q & 32.56 & 35.54 & 49.30 & 43.74 & 53.36 & 58.54 & 46.78 \\ 
  r & 27.27 & 31.39 & 48.86 & 44.20 & 52.92 & 57.78 & 46.56 \\ 
  \hline
  mean & 14.92 & 20.98 & 45.58 & 43.10 & 50.32 & 55.43 & 28.86 \\ 
   \hline
\end{tabular}}
\caption{Amari error ($\times 100$) for $d=6$, 
   multiplicative contamination.}
\end{table}

\end{document}